\documentclass[11pt]{article}
			
\usepackage{fullpage} 
\usepackage{algorithm,algorithmic}
\usepackage{amsmath,amssymb,amsthm,amsfonts}
\usepackage{booktabs}   
\usepackage{xspace}
\usepackage{hyperref,url}
\usepackage{graphicx,subfigure}
\usepackage{mdframed}
\usepackage{thm-restate}
\definecolor{light-gray}{gray}{0.95}
\usepackage{mathtools}
\usepackage{tabu,multirow}
\usepackage{dblfloatfix}
 \usepackage{epsfig}
\usepackage{xcolor,color,url,eurosym,bbm}
\usepackage{array,multirow,tikz,wrapfig}
\usetikzlibrary{shapes.geometric}
\usepackage{empheq}
\usepackage{adjustbox}
\usepackage[graphicx]{realboxes}

\usepackage[most]{tcolorbox}
\newtcbox{\mymath}[1][]{%
    nobeforeafter, math upper, tcbox raise base,
    enhanced, colframe=blue!30!black,
    colback=blue!30, boxrule=1pt,
    #1}

\tikzset{base/.style={draw, align=center, minimum height=4ex},
         test1/.style={base, diamond, aspect=2, text width=5em, inner sep=5pt},
         test2/.style={base, diamond, aspect=2, text width=5em, inner sep=-4.8pt}
        }

 \newcommand{\specialcell}[2][c]{%
  \begin{tabular}[#1]{@{}c@{}}#2\end{tabular}}

\newcommand{\xuv}{\chi_{u,v}}
\newtheorem{problem}{Problem}
\newcommand\myeq{\mathrel{\stackrel{\makebox[0pt]{\mbox{\normalfont\tiny def}}}{=}}}
\newtheorem{proposition}{Proposition}
\newtheorem{theorem}{Theorem}
\newcommand{\diag}{\mathop\mathrm{diag}}

\DeclareMathOperator{\tr}{Tr}

\DeclareMathOperator{\vol}{vol}

 \makeatletter
\newcommand*{\rom}[1]{\expandafter\@slowromancap\romannumeral #1@}
\makeatother

\newcommand{\hide}[1]{} 

\newtheorem{definition}{Definition}

\newcommand{\eqdef}{\mathbin{\stackrel{\rm def}{=}}}

\newcommand{\rg}{r_G}
\newcommand{\pp}{p^{\alpha}}
\newcommand{\pone}{\hyperref[fullProb]{1.1}}
\newcommand{\ptwo}{\hyperref[subProb]{1.2}}
\newcommand{\pthree}{\hyperref[noiseProb]{1.3}}

\newcommand{\wh}{\widehat}

\newcommand{\ol}{\overline}

\newcommand{\optG}{H}

\newcommand{\beql}[1]{\begin{equation}\label{#1}}

\newcommand{\beq}[1]{\begin{equation}\label{#1}}
\newcommand{\eeq}{\end{equation}}

\newcommand{\field}[1]{\mathbb{#1}} 

\newcommand{\spara}[1]{\smallskip\noindent{\bf #1}}

\newcommand{\squishlist}{
 \begin{list}{$\bullet$}
  {  \setlength{\itemsep}{0pt}
     \setlength{\parsep}{3pt}
     \setlength{\topsep}{3pt}
     \setlength{\partopsep}{0pt}
     \setlength{\leftmargin}{2em}
     \setlength{\labelwidth}{1.5em}
     \setlength{\labelsep}{0.5em}
} }
\newcommand{\squishlisttight}{
 \begin{list}{$\bullet$}
  { \setlength{\itemsep}{0pt}
    \setlength{\parsep}{0pt}
    \setlength{\topsep}{0pt}
    \setlength{\partopsep}{0pt}
    \setlength{\leftmargin}{2em}
    \setlength{\labelwidth}{1.5em}
    \setlength{\labelsep}{0.5em}
} }

\newcommand{\squishdesc}{
 \begin{list}{}
  {  \setlength{\itemsep}{0pt}
     \setlength{\parsep}{3pt}
     \setlength{\topsep}{3pt}
     \setlength{\partopsep}{0pt}
     \setlength{\leftmargin}{1em}
     \setlength{\labelwidth}{1.5em}
     \setlength{\labelsep}{0.5em}
} }

\newcommand{\squishend}{
  \end{list}
}

\newcommand{\squishlistt}{
 \begin{list}{---}
  {  \setlength{\itemsep}{0pt}
     \setlength{\parsep}{3pt}
     \setlength{\topsep}{3pt}
     \setlength{\partopsep}{0pt}
     \setlength{\leftmargin}{2em}
     \setlength{\labelwidth}{1.5em}
     \setlength{\labelsep}{0.5em}
} }

\begin{document}

\title{Learning Networks from Random Walk-Based Node Similarities}

\author{
Jeremy G. Hoskins\thanks{Yale University. \href{jeremy.hoskins@yale.edu}{jeremy.hoskins@yale.edu} }
\and 
Cameron Musco\thanks{Massachusetts Institute of Technology. \href{cnmusco@mit.edu}{cnmusco@mit.edu} }
\and 
Christopher Musco\thanks{Massachusetts Institute of Technology. \href{cpmusco@mit.edu}{cpmusco@mit.edu} }
\and
Charalampos E. Tsourakakis\thanks{Boston University \&  Harvard University. \href{mailto:ctsourak@bu.edu}{ctsourak@bu.edu} }
}

\maketitle 
 
\begin{abstract}
Digital presence in the world of online social media entails significant privacy risks \cite{backstrom2007wherefore,bhagat2009class,korolova2008link,zheleva2012privacy,korula2014efficient}. In this work we consider a privacy threat to a social network in which an attacker has access to a subset of \emph{random walk-based node similarities}, such as effective resistances (i.e., commute times) or personalized PageRank scores. Using these similarities, the attacker's goal is to infer as much information as possible about the underlying network, including any remaining unknown pairwise node similarities and edges.

For the effective resistance metric, we show that with just a small subset of measurements, the attacker can learn a large fraction of edges in a social network (and in some cases all edges), even when the measurements are noisy. We also show that it is possible to  learn a graph which accurately matches the underlying network on \emph{all other effective resistances.} This second observation is interesting from a data mining perspective, since it can be expensive to accurately compute all effective resistances or other random walk-based similarities. As an alternative, our graphs learned from just a subset of approximate effective resistances can be used as surrogates in a wide range of applications that use effective resistances to probe graph structure, including for graph clustering, node centrality evaluation, and anomaly detection. 

We obtain our results by formalizing the graph learning objective mathematically, using two optimization problems. One formulation is convex and can be solved provably in polynomial time. The other is not, but we solve it efficiently with projected gradient and coordinate descent. 
We demonstrate the effectiveness of these methods on a number of social networks obtained from Facebook.
We also discuss how our methods can be generalized to other random walk-based similarities, such as personalized PageRank scores. Our code is available at\linebreak \url{https://github.com/cnmusco/graph-similarity-learning}.

\end{abstract}

\newpage

\section{Introduction}
\label{sec:intro}
In graph mining and social network science, a variety of measures are used to quantify the similarity between nodes in a graph, including the shortest path distance, Jaccard's coefficient between node neighborhoods, the Adamic-Adar coefficient \cite{adamic2003friends}, and hub-authority-based metrics \cite{kleinberg1999web,blondel2004measure}.
An important family of similarity measures are based on random walks, including SimRank \cite{jeh2002simrank},  random walks with restarts \cite{tong2006fast},  commute times \cite{fouss2007random},  personalized PageRank  \cite{page1999pagerank,jeh2003scaling,bahmani2010fast}, and DeepWalk embeddings \cite{perozzi2014deepwalk}. These measures capture both local and global graph structure and hence are widely used in graph clustering and community detection \cite{andersen2006local,saerens2004principal}, anomaly detection \cite{rattigan2005case}, collaborative filtering \cite{fouss2007random,sarkar2012tractable,yen2007graph}, link prediction \cite{liben2007link}, and many other applications (including outside of network science, such as in computer vision \cite{grady2006random}).

In this work we focus on these random walk-based similarity metrics. We initiate the  study of a fundamental question:
\begin{quotation}
How much information about a network can be learned given access to \emph{a subset of potentially noisy estimates} of pairwise node similarities?
\end{quotation}

This question is important from a privacy perspective. 
A common privacy breach is {\em social link disclosure} \cite{zheleva2012privacy}, in which an attacker attempts to learn potentially sensitive links between nodes in a network. Such attacks are very common; fake accounts with engineered profiles are used to infiltrate and spy on social groups,  potential employers  may want to inspect the social network of a job candidate, and advertisers may wish to probe the demographic and interest information of a user to offer targeted ads. Thus, characterizing the ability of an attacker to reveal link information using pairwise node similarities is important in understanding the privacy implications of releasing such similarities, or information which can be used to compute them.

From a data mining perspective, computing all pairwise node similarities can be infeasible for large networks since the number of similarities grows quadratically  in the number of nodes. Additionally, when the network cannot be accessed in full but can only be probed via crawling with random walks (e.g., a by third party studying a social network \cite{katzir2011estimating}), we may only have access to estimates of pairwise similarities rather than their exact values. Thus, understanding what information can still be learned from a partial, potentially noisy, set of node similarities is important when using these metrics in large scale graph mining.

\vspace{-3mm}
\subsection{Learning from Effective Resistances}
\label{intro_learning}

In this paper, we focus on commute times, which are one of the most widely used random walk-based similarities.
Commute times are a scaled version of \emph{effective resistances}, they form a metric, and have major algorithmic applications, such as spectral graph sparsification \cite{spielman2011graph}.  
Our ideas can be extended to related similarity measures, such as personalized PageRank, which we discuss in Section \ref{sec:extensions}.
It was shown in the seminal work of Liben-Nowell and Kleinberg  \cite{liben2007link} that effective resistances can be used to predict a significant fraction of future links appearing in networks from existing links,  typically ranging  from 5\% up to 33\%.  

A difficulty associated with this task is that, in contrast 
to {\em local} similarity measures such as the number of common neighbors or the Adamic-Adar coefficient \cite{adamic2003friends}, node similarity under the effective resistance metric does not necessarily imply local connectivity. For example, two nodes connected by many long paths may be more similar than two nodes directly connected by a single edge.  

Furthermore, in certain cases, the effective resistance between two nodes $u,v$ tends to correlate well with a simple function of the degree sequence (specifically, $\frac{1}{d(u)}+\frac{1}{d(v)}$) \cite{luxburg2010getting,von2010hitting,von2014hitting}, and it is known that there  are many graphs with the same degree sequence but very different global structures.
Nevertheless, \emph{considered in aggregate}, effective resistances encode global structure in a very strong way. For any graph, given all pairs effective resistances, it is possible to provably recover the full graph in polynomial time \cite{spielmanLecture,wittmann}! This contrasts with purely local similarity metrics, which can be used heuristically for link prediction, but do not give network reconstruction in general. For instance,  all-pairwise counts of common neighbors in any triangle free graph equal $0$. Thus, these counts reveal no information about graph structure.

While the full information case is well understood, when all exact effective resistances are not available, little is known about what graph information can be learned. Some work considers reconstruction of trees based on a subset of effective resistances \cite{CULBERSON1989215,Batagelj90,stone2009}. However outside of this special case, essentially nothing is known.

\begin{figure*}[!t]
\centering
\includegraphics[width=.85\textwidth]{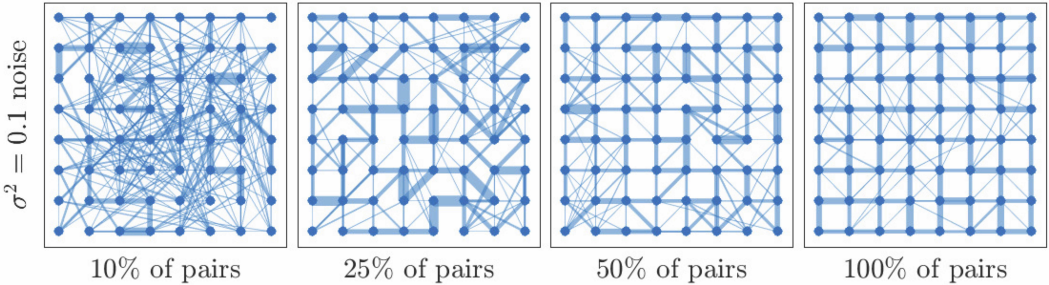}
\caption{\label{teaserimage}Grid graph learned from  small, randomly selected fractions of effective resistance pairs using our proposed method. Edge thickness is proportional to edge weight.
}
\end{figure*}

\subsection{Our Contributions}

We study in depth what can be learned about a graph given a subset of potentially noisy effective resistance estimates, from both a theoretical and empirical viewpoint. Our main contributions are:

\spara{Mathematical formulation.}  
We provide an optimization-based formulation of the problem of learning a graph from effective resistances. Specifically, given a set of effective resistance measurements, we consider the problem of finding a graph whose effective resistances \emph{match the given resistances as closely as possible.} 

In general, there may be many different graphs which match any subset of all pairs effective resistances, and hence many minimizers to our optimization problem. If the resistances additionally have some noise, there may be no graph which matches them exactly but many which match them approximately.
 Nevertheless, as we show empirically, the graph obtained via our optimization approach typically recovers significant information about the underlying graph, including a large fraction of its edges, its global structure, and good approximations to \emph{all} of its effective resistances. 

\spara{Algorithms.}  
We prove that, in some cases, the optimization problem we present can be solved exactly, in polynomial time. However, in general, the problem is non-convex and does not admit an obvious polynomial time solution. We give expressions for its gradient and Hessian, and show that it can be solved efficiently via iterative methods. In particular, we employ projected gradient and coordinate descent, as well as a powerful initialization strategy which allows us to find high quality solutions in most instances. 

We also show that the problem can be relaxed to a convex formulation. Instead of searching for a graph that matches all given effective resistance measurements, we just find a graph whose effective resistances are upper bounded by those given and which has minimum total edge weight. This modified problem is convex and can be solved in polynomial time via semidefinite programming.

\spara{Experimental Results.} 
We evaluate our algorithms on several synthetic graphs and real Facebook ego networks, which contain all nodes in the social circle of a user. We demonstrate that given a small randomly selected fraction of all effective resistance pairs ($10\%$-$25\%$) we can learn a large fraction of a network -- typically between $20\%$ and $60\%$ of edges. We show that this recovery  is robust to adding Gaussian noise to the given effective resistances. A visual sample of our results is included in Figure \ref{teaserimage} for a simple grid graph.  We observe that using a small fraction of noisy pairwise similarities we can extract a lot of structural information about the graph. Here, noise is distributed for each edge independently as a Gaussian random variable with zero mean, and variance $\sigma^2=0.1$.  We observe that as the size of the fraction of revealed similarities increases, an attacker can recover most of the graph. 

We also demonstrate that by finding a graph which closely matches the given set of effective resistances (via our optimization approach), we in fact find a graph which closely matches the underlying network on \emph{all} effective resistance pairs. This indicates that a significant amount of the information contained in all pairs effective resistances can be learned from just a small subset of these pairs, even when they are corrupted by noise.


\section{Related Work}
\label{sec:related}
\spara{Link prediction and privacy in social networks.}  The link prediction problem was popularized by Liben-Nowell and Kleinberg \cite{liben2007link}. The goal of {\em link prediction} is to infer which edges are likely to appear in the near-future, given a snapshot of the network. In \cite{liben2007link} various node similarity measures are used to predict a non-trivial fraction of future interactions. Other works focus on predicting positive or negative links \cite{leskovec2010predicting,tsourakakis2017predicting}.  For an extensive survey, see \cite{al2011survey}. 

While closely related, link prediction differs from the problem we consider since, typically, one is given full access to a network snapshot (and in particular, could compute exact pairwise node similarities for this network) and seeks to make predictions about future evolutions of the network. In our setting, we are given partial information about a network (via a partial set of noisy similarity measures) and seek to learn existing links, 

While link prediction is useful in applications ranging from understanding network evolution, to link recommendation, and predicting interactions between terrorists, it entails privacy risks.  For instance, a malicious attacker can use link prediction to disclose information about sensitive connections in a sexual contact graph   \cite{korolova2008link,zheleva2012privacy}.   Private link prediction is not adequately explored in the literature. Abebe and Nakos suggest a possible formalization \cite{abebe2014private}.

\spara{Learning graphs.} Learning graphs from data is central to many disciplines including machine learning \cite{kalofolias2016learn}, network tomography \cite{castro2004network}, bioinformatics \cite{desper1999inferring}, and phylogenetics \cite{felsenstein1985confidence}. 
 The general framework is that there exists a hidden graph that we wish to discover by exploiting some kind of data, e.g., answers from a blackbox oracle returning certain (possibly noisy) graph measurements. 
 
Theoretical work in this area has focused on worst case {\em query} complexity. Two representative examples include Angluin's et al. work on learning graphs using edge detecting queries \cite{angluin2008learning}, and the recent work of Kannan, Mathieu, and Zhou using distance queries \cite{kannan2015near}. A number of works also consider learning graph parameters such as node count and mixing time by examining random walks traces \cite{katzir2011estimating,cooper2014estimating,ben2017estimating}. Experimental work has focused on recovering graphs from  noisy measurements such as GPS traces \cite{chen2010road}, and distances between cell populations based on genetic differences \cite{felsenstein1985confidence,desper1999inferring}.

\spara{Learning trees.} While a tree is special type of a graph, the case of learning trees deserves special mention since significant work on learning graphs from data has focused on trees.  Distance-based reconstruction of trees aims to reconstruct a phylogenetic tree whose leaves correspond to $n$ species, given their ${n \choose 2}$ pairwise distances \cite{felsenstein1985confidence}. Note that on a tree, pairwise distances are identical to pairwise effective resistanes. Batagelj et al. study tree realizability assuming access to fewer than ${n \choose 2}$ leaves' pairwise distances \cite{Batagelj90}.   A spectral method has been proposed by Stone and Griffing \cite{stone2009}.  Culberson and Rudnicki  \cite{CULBERSON1989215} consider the problem of reconstructing a degree restricted tree
given its distance matrix, see also \cite{reyzin2007longest}.

\section{Proposed Method}
\label{sec:proposed}
\subsection{Notation and Preliminaries}
For an undirected, weighted graph $G = (V,E,w)$ with $n$ nodes, we let $A$ be the $n\times n$ adjacency matrix. $L$ denotes the (unnormalized) graph Laplacian:
$L = D - A$, where $D$ is a diagonal matrix with $D_{i,i}$ equal to the weighted degree of node $i$. For an integer $n > 0$, $[n]$ denotes the set $\{1,2,...,n\}$. $e_i$ denotes the $i^{th}$ standard basis vector. For a matrix $M$, $M_{i,j}$ denotes the entry in its $i^{th}$ row and $j^{th}$ column.

\spara{Commute time and effective resistance.}
For two nodes $u,v \in V$, the hitting time $h_G(u,v)$ is the expected time it takes a random walk to travel from $u$ to $v$. The {\em commute time} is its symmetrized version $c_G(u,v) = h_G(u,v) + h_G(v,u)$, i.e., the time to move from $u$ to $v$ and then back to $u$. 
For connected graphs, the effective resistance between $u,v$ is a scaling of the commute time: $r_G(u,v) =\frac{c_G(u,v)}{\vol(G)}$ where $\vol(G) = 2 \sum_{e \in E} w_e$. Effective resistance has a natural electrical interpretation. When $G$ is viewed as an electrical network on $n$ nodes where each edge $e$ corresponds to a link of conductance $w_e$ (equivalently to a resistor of resistance $\frac{1}{w_e}$), the effective resistance is the voltage difference that appears across  $u,v$  when a unit current source is applied to them. Effective resistances (and hence commute times) always form a metric \cite{klein1993resistance}.

Letting $\xuv = e_u - e_v$, the effective resistance between nodes $u$ and $v$ in a graph $G$ with Laplacian matrix $L$ can be computed as:
\begin{align}\label{eq:res}
\rg(u,v) = \xuv^TL^+ \xuv.
\end{align}
Here $L^+$ denotes the Moore-Penrose pseudoinverse of $L$.

\subsection{Problem Definition} 

We begin by providing a mathematic formulation of the problem introduced in Section~\ref{sec:intro} -- that of learning the structure of a graph from partial and possibly noisy measurements of pairwise effective resistances. An analogous problem can be defined for other random walk-based similarities, such as personalized PageRank. We discuss initial results in this direction in Section \ref{sec:extensions}. 
\clearpage
\vspace{1.5em}
\begin{mdframed}[backgroundcolor=light-gray]
\vspace{-.5em}
\begin{problem}
[Graph Reconstruction From Effective Resistances]\label{prob:effres}
Reconstruct an unknown graph $G$ given a set of noisy effective resistance measurements, $$\bar r(u,v) = {r}_G(u,v) + n_{uv}$$
for each $(u,v) \in \mathcal{S}$, where $\mathcal{S} \subseteq [n] \times [n]$ is a set of node pairs and $n_{uv}$ is a 	potentially random noise term.
\end{problem}
\end{mdframed}
\vspace{.5em}

We focus on three interesting cases of Problem \ref{prob:effres}:
\begin{description}
\item[Problem 1.1]\label{fullProb} $\mathcal{S} = [n] \times [n]$ and $n_{uv} = 0$ for all $(u,v) \in S$. This is the \emph{full information recovery problem.}
\item[Problem 1.2]\label{subProb} $\mathcal{S}$ is a subset of $[n] \times [n]$ and $n_{uv} = 0$ for all $(u,v) \in \mathcal{S}$. 
In this setting we must learn $G$ from a limited number of exact effective resistances.
\item[Problem 1.3]\label{noiseProb} $\mathcal{S}$ is a subset of $[n] \times [n]$ and $n_{uv}$ is a random term, e.g. a mean $0$ normal random variable with variance $\sigma^2$: $n_{uv} \sim \mathcal{N}(0,\sigma^2)$. 
\end{description}

It is known that there in a unique graph consistent with any full set of effective resistance measurements (see e.g., \cite{spielmanLecture} or the related problem in \cite{wittmann}). Additionally, this graph can be computed by solving a fully determined linear system. So, we can solve Problem \pone\ exactly in polynomial time. We illustrate this in  Section~\ref{sec:full_reconstruct}.


From a privacy and data mining perspective, the limited information settings of Problems \ptwo\ and \pthree\ are more interesting. In Section \ref{sec:full_reconstruct} we demonstrate that, when $G$ is a tree, exact recovery is possible for Problem \ptwo\ whenever $\mathcal{S}$ is a superset of $G$'s edges. However, in general, there  is no simple closed form solution to these problems, and exact recovery of $G$ is typically impossible. In particular, several graphs may be consistent with the measurements given.
Thus, we address these cases by reposing Problem \ref{prob:effres} as an optimization problem, in which we attempt to recover a graph matching the given effective resistances as best as possible.

\subsection{Optimization Formulation} 

A natural formalization of Problem \ref{prob:effres} is as a least squares problem. 

\vspace{0.5em}
\begin{mdframed}[backgroundcolor=light-gray] 
\vspace{-.5em}
\begin{problem}
\label{prob:effres2} Given a set of vertex pairs $\mathcal{S} \subseteq [n] \times [n]$ and a target effective resistance $\bar r(u,v)$ for each $(u,v) \in \mathcal{S}$: 
\begin{equation}\label{eq:reg}
\begin{aligned}
& \underset{\text{graph } \optG}{\text{minimize}} 
& & F(\optG) \eqdef \sum_{(u,v) \in \mathcal{S}} \left [{r}_\optG(u,v)-\bar r(u,v)\right]^2.
\end{aligned}
\end{equation}
\end{problem}
\end{mdframed}
\vspace{0.5em}

Using formula \eqref{eq:res} for effective resistances, 
Problem \ref{prob:effres2} can equivalently be viewed as an optimization problem over the set of graph Laplacians: we minimize $\sum_{(u,v)\in \mathcal{S}} \left [\xuv^T L^+ \xuv - \bar r(u,v)  \right ]^2$. While this set is convex, the objective function is not and thus it is unclear if it can be minimized provably  in polynomial time.
Nevertheless, we demonstrate that it is possible to solve the problem approximately using iterative methods. In Section \ref{sec:full_reconstruct_noise} we derive projected gradient and coordinate descent algorithms for the problem. Combined with a powerful initialization heuristic, our experiments show that these methods quickly converge to near global minimums of the objective function for many networks.

For Problem \ptwo, where $\bar r(u,v)$ comprise a subset of the exact effective resistances for some graph $G$, $\min_\optG F(\optG) = 0$. This minimum may be achieved by multiple graphs (including $G$ itself) if $\mathcal{S}$ does not contain all effective resistance pairs. Nevertheless, we demonstrate experimentally in Section \ref{sec:exp} that even when $\mathcal{S}$ contains a small fraction of these pairs, an approximate solution to Problem \ref{prob:effres2} often recovers significant information about $G$, including a large fraction of its edges. Interestingly, we find that, while Problem \ref{prob:effres2} only minimizes over the subset $\mathcal{S}$, the recovered graph typically matches $G$ on \emph{all} effective resistances, explaining why it contains so much structural information about $G$. For Problem \pthree, if $\mathcal{S} = [n] \times [n]$ and the noise terms $n_{uv}$ are distributed as i.i.d. Gaussians, it is not hard to see that Problem \ref{prob:effres2} gives the maximum likelihood estimator for $G$. We again show that an approximate solution can recover a large fraction of $G$'s edges. 

We note that while we can solve Problem \ref{prob:effres2} quickly via iterative methods, we leave open provable polynomial time algorithms for solving this problem in the settings of both Problems \ptwo\ and \pthree.

\spara{Convex relaxation.}
As an alternative to Problem \ref{prob:effres2}, we introduce an optimization formulation of Problem \ref{prob:effres} that \emph{is convex}. It is convenient here to optimize  over the convex set of graph Laplacians.

\vspace{.5em}
\begin{mdframed}[backgroundcolor=light-gray] 
\vspace{-.5em}
\begin{problem}
\label{prob:effres3} 
Let $\mathcal{L}$ be the convex set of $n \times n$ graph Laplacians.
Given a set of vertex pairs $\mathcal{S} \subseteq [n] \times [n]$ and a target effective resistance $\bar r_{(u,v)}$ for every $( u,v ) \in \mathcal{S}$,
\begin{equation*}
\begin{aligned}
& \underset{L \in \mathcal{L}}{\text{minimize}}
& & \tr(L) \\
& \text{subject to} & \quad \xuv^T L^+ \xuv \leq \bar r(u,v) & \text{~~~} \forall (u,v) \in \mathcal{S} \\ 
\end{aligned}
\end{equation*}
\end{problem}
\end{mdframed}


Observe that we can trivially find a feasible solution for Problem \ref{prob:effres3} by setting $L$ to be a  large complete graph, in which case all effective resistances will be close to 0.  By Rayleigh's monotonicity law, decreasing the weight on edges in $L$ increases effective resistances. $\tr(L)$ is equal to the total degree of the graph corresponding to $L$, so the problem asks us to find a graph with as little total edge weight as possible that still satisfies the effective resistance constraints.

The disadvantage of this formulation is that it only encodes the target resistances $\bar r(u,v)$ as \emph{upper bounds} on the resistances of $L$. The advantage is that we can solve Problem \ref{prob:effres3} provably in polynomial time via semidefinite programming (see Section \ref{sec:partial_reconstruct_noise}). In practice, we find that it can sometimes effectively learn graph edges and structure from limited measurements.

Problem \ref{prob:effres3} is related to  work on convex methods for minimizing total effective resistance or relatedly, mixing time in graphs \cite{boyd2004fastest,sun2006fastest,ghosh2008minimizing,ghosh2008minimizing}. However, prior work does not consider pairwise resistance constraints and so is not suited to the graph learning problem. 

\vspace{3mm} 


\section{Analytical Results and Algorithms}\label{sec:analytical}

\subsection{Full Graph Reconstruction -- Problem \ref{prob:effres}}  
 \label{sec:full_reconstruct}

Problem \ref{prob:effres} can be solved exactly in polynomial time when $\mathcal{S}$ contains \emph{all resistance pairs} of some graph $G$ (i.e. Problem \pone). In this case, there is a closed form solution for $G$'s Laplacian $L$ and the solution is unique. This was pointed out in \cite{spielmanLecture}, however we include our own proof for completeness.

\begin{theorem}
\label{thrm:profeffres} 
If there is a feasible solution to Problem \pone\, then it is unique and can be found in $O(n^3)$ time. Specifically, the Laplacian $L$ of the unique solution $G$ is given by 
\begin{align}\label{eq:exactSolution}
-2 \cdot \left [ \left  (I - \frac{J}{n}\right)R\left(I-\frac{J}{n}\right)\right ]^+
\end{align}
where $R$ is the matrix with  $R_{u,v} = r_G(u,v)$ for all $u,v \in [n]$, $I$ is the $n \times n$ identity matrix , and $J$ is the $n \times n$ all ones matrix.
\end{theorem}
\begin{proof}
For an $n$ node graph, the number of possible edges is $m = {n \choose 2}$. Let $B \in \mathbb{R}^{m \times n}$ be the vertex edge incidence matrix of the complete graph with a row equal to $\xuv = e_u- e_v$ for every $(u,v)$.

The Laplacian of any $n$ vertex graph can be written as $L = B^T W B$, for some $W \in \mathbb{R}^{m \times m}$ which is a nonnegative diagonal matrix with entries corresponding to the edge weights.

We can rewrite the effective resistance formula in \eqref{eq:res} as:
\begin{align}\label{eq:effresRewrite}
r_L(u,v) = \xuv^\top L^+ \xuv = (L^+)_{u,u} + (L^+)_{v,v} - 2(L^+)_{u,v}.
\end{align}


Since $L^+$ is symmetric we need only determine $\frac{n(n+1)}{2}$ unknown entries to determine the full matrix. Moreover, since the all ones vector is in the null space of $L$ and therefore $L^+$, we see that:
\begin{align}\label{eq:diag}
L^+_{u,u} = -\sum_{v \neq u} L^+_{u,v},
\end{align}
and hence we can rewrite \eqref{eq:effresRewrite} as:
\begin{align}\label{eq:resistance_system}
r_{u,v}  = -\sum_{v' \neq v, u} (L^+)_{u,v'}-\sum_{u' \neq u,v} (L^+)_{u',v} - 4(L^+)_{u,v}.
\end{align}

Let $M$ be the $m \times m$ matrix with rows and columns indexed by pairs $u,v \in [n]$ with $u \neq v$ and $(u_1,v_1), (u_2,v_2)$ entry given by:
\begin{align*}
M_{(u_1,v_1), (u_2,v_2)} =\begin{cases}
-4 \text{ if } u_1=u_2\text{ and }v_1=v_2\\
-1\text{ if } u_1 = u_2\text{ or }v_1 = v_2\\
0\text{ otherwise}.
\end{cases}
\end{align*}
Let $r \in \field{R}^m$ contain each effective resistance $r_G(u,v)$.
We can see from \eqref{eq:resistance_system} that if we solve the linear system $M x = -r$, as long as $M$ is full rank and so the solution is unique, the entries of $x$ will give us each $(L^+)_{u,v}$ with $u \neq v$. We can then use these entries to recover the remaining diagonal entries of $L^+$ using \eqref{eq:diag}.

We can verify that $M$ is in fact always full rank  by  writing $M = -|B| |B|^T-2I,$ where $|B|$ denotes the matrix formed from $B$ by taking the absolute value of each of its entries. We note that the non-zero eigenvalues of $|B| |B|^T$ are equal to the non-zero eigenvalues of $|B^T||B|$ which is the $n \times n$ matrix $\ol{M} = (n-2)I + J $. $\ol{M}$ has eigenvalues $2n-2$ with multiplicity $1$ and $n-2$ with multiplicity $n-1.$ The remaining $m-n$ eigenvalues of $|B| |B|^T$ are zero. Consequently, the eigenvalues of $M$ are $-2n$ with multiplicity $1,$ $-n$ with multiplicity $n-1$ and $-2$ with multiplicity $m-n.$ Thus $M$ is full rank, proving that the effective resistances fully determine $L^+$ and thus $L$. 

Solving for $L$  via  the linear system $Mx = -r$ would require $O(n^6)$ time, however,  the closed form solution \eqref{eq:exactSolution} given in Lemma 9.4.1 of \cite{spielmanLecture} allows us to solve this problem in $O(n^3)$ time.
\end{proof}

\spara{Reconstruction from hitting times.} The above immediately generalizes to graph reconstruction from hitting times since, as discussed, for connected $G$, the effective resistance between $u,v$ can be written as $r_G(u,v) = \frac{c_G(u,v)}{vol(G)} = \frac{h_G(u,v) +  h_G(v,u)}{vol(G)}$. Thus, by  Theorem \ref{thrm:profeffres}, we can recover $G$ up to a scaling from all pairs hitting times. This recovers a result in \cite{wittmann}. Note that  if we scale all edge weights in $G$ by a  fixed factor, the hitting times do not  change. Thus recovery up to a scaling is the best we can hope for in this setting.

\spara{Reconstruction from other similarity  measures.} 
An analogous result to Theorem \ref{thrm:profeffres} holds for graph recovery from all pairs personalized PageRank scores, and for related measures such as Katz similarity scores \cite{katz1953new}. We discuss this direction in Section \ref{sec:extensions}.

%

\spara{Are all pairs always necessary for perfect reconstruction?}  
For general graphs, Problem \ref{prob:effres} can only be solved exactly when $\mathcal{S}$ contains all ${n \choose 2}$ true effective resistances. However, given additional constraints on $G$, recovery is possible with much less information.
In particular, when $G$ is a tree, we can recover it (i.e., solve Problem \ptwo) whenever $\mathcal{S}$ is a superset of its edge set.

Roughly, since $G$ is a tree, the effective resistance $\rg(u,v)$ is equal to the length of the unique path connecting $u$ and $v$. As long as $\mathcal{S}$ includes all edges in $G$, it fully determines all path lengths and hence the effective resistances for all pairs $u,v$. We can thus recover $G$ via Theorem \ref{thrm:profeffres}. Formally:

\begin{theorem}\label{thm:tree}
If $G$ is a tree and a feasible solution to Problem \ptwo\ with edge set $E \subseteq \mathcal{S}$ then $G$ is unique and can be found in $O(n^3)$ time.
\end{theorem}
\begin{proof}
Let $P_{uv}$ be the unique path between $u,v$ in $G$. It is well known \cite{chandra1996electrical} that:
\begin{align}\label{eq:treeEffRes}
\rg(u,v) = \sum_{e \in P_{uv}} 1/w_e.
\end{align}
For $(u,v) \in \mathcal{S}$ set $\bar r(u,v) = \rg(u,v)$.  Let  $\ol{G}$ be an undirected graph with an edge for each $(u,v) \in \mathcal{S}$ with length $\rg(u,v)$. For all $(u,v) \notin \mathcal{S}$, set $\bar r(u,v)$ to the shortest path distance between $u$ and $v$ in  $\ol{G}$.

\spara{Claim.} $\bar r(u,v) = \rg(u,v)$ for all $(u,v) \in [n] \times [n]$.\vspace{.5em}\\
For \emph{any pair} $(u,v)$, we have $\bar r(u,v) \le \rg(u,v)$. The length of shortest path between $u,v$ in $\ol{G}$ is certainly at most the length of $P_{uv}$, which is contained in $\ol{G}$ since $E \subseteq \mathcal{S}$. $P_{uv}$'s length in $\ol G$ is: $\sum_{e \in P_{uv}} r_G(e) = \sum_{e \in P_{uv}} 1/w_e$ using \eqref{eq:treeEffRes}. Thus, $\bar r(u,v) \le \sum_{e \in P_{uv}} 1/w_e = r_G(u,v)$, again using \eqref{eq:treeEffRes}.

Further, $P_{uv}$ is in fact a shortest path between $u,v$ in $\ol{G}$, giving that $\bar r(u,v) = \rg(u,v)$. This is because the length of every edge $(u,v) \in \mathcal{S}$ that is not in $E$ just equals the length of path $P_{uv}$ in $\ol{G}$ (i.e., $r_G(u,v)$) and so removing this edge from $\ol G$ does not change any shortest path distance. So we can assume that $\ol{G}$  just contains the edges in $E$, and so $P_{uv}$ is the unique path between $u,v$.
\vspace{.5em}

Given the above claim, the theorem follows since we can compute each $\bar r(u,v)$ from the effective resistances of the edges in $\mathcal{S}$ and can then compute $G$ from these resistances by Theorem  \ref{thrm:profeffres}.
\end{proof} 

The problem of recovering trees from pairwise distance measurements is a central problem in phylogenetics. There are other cases when just a subset of effective resistances is known to allow full recovery, for example when the effective resistances between any pair of leaves is known \cite{stone2009,spielmanLecture}. Also related to our analysis for trees lies is the work of Mathieu and Zhou \cite{mathieu2013graph}, which reconstructs graphs with bounded degree from pairwise distance measurements. 
 
\subsection{Graph Learning via Least Squares Minimization -- Problem~\ref{prob:effres2} }  
\label{sec:full_reconstruct_noise} 

When Problem \ref{prob:effres} cannot be solved exactly, e.g. in the settings of  Problems \ptwo\ and \pthree, an effective surrogate is to solve Problem \ref{prob:effres2} to find a graph with effective resistances close to the given target resistances. As we demonstrate experimentally in Section \ref{sec:exp}, this yields good solutions to Problems \ptwo\ and \pthree\ in many cases. Problem \ref{prob:effres2} is non-convex, however we show that a good solution can often be found efficiently via projected gradient descent. 

\spara{Optimizing over edge weights.}
Let $m = {n \choose 2}$. We write the Laplacian of the graph $H$ as $L(w) \eqdef B^T \diag(w) B$, where $w \in \mathbb{R}^{m}$ is a non-negative vector  whose entries correspond  to the edge weights in $H$, $\diag(w)$ is the $m \times m$ matrix with $w$  as its diagonal, and $B \in\mathbb{R}^{m \times n}$ is the vertex edge incidence matrix with a row equal to $\xuv = e_u-e_v$ for every possible edge $(u,v) \in [n] \times [n]$.

Optimizing the objective function $F(H)$ in Problem \ref{prob:effres2} is equivalent to optimizing $F(w)$ over the edge weight vector $w$, where we define $F(w) \eqdef F(H)$ for the unique $H$ with Laplacian equal to $L(w)$.

We restrict $w_i \ge 0$ for all $i$ and project to this constraint after each gradient step simply by setting $w_i := \max(w_i,0)$. The gradient of $F(w)$ can be computed in closed form. We first define an auxiliary variable, $R(w) \in \mathbb{R}^{m \times m}$, whose diagonal contains all pairwise effective resistances of $H$ with weight vector $w$:
\begin{definition}
For $w \in \mathbb{R}^m$ with $w_i \ge 0$ for all $i$, define
\begin{align*}
R(w)  = B L(w)^+ B^T.
\end{align*} 
\end{definition}
  
Using $R(w)$ we can compute the gradient of $F(w)$ by:
\begin{proposition}
\label{prop:gradient}
Let $\circ$ denote the Hadamard (entrywise) product for matrices. 
Define the error vector $\Delta(w) \in \mathbb{R}^m$ as having $\Delta(w)_{i} = \bar r(i) - [R(w)]_{i,i}$ for all $i \in \mathcal{S}$ and $0$s elsewhere. We have:
\begin{align*}
\nabla F(w)= 2\left(R\circ R\right) \Delta(w)  
\end{align*} 
\end{proposition}
\begin{proof}
We begin by observing that, letting $e_i$ be the $i$th standard basis vector in $\mathbb{R}^m$, for any weight vector $w,$ the graph Laplacian corresponding to the weight vector $w +\epsilon e_i,$ is $L(w)+\epsilon b_i b_i^T.$ The Sherman-Morrison formula for the matrix pseudoinverse yields:
\begin{align*}
(L(w)+\epsilon b_ib_i^T)^+ = L(w)^+ -\epsilon \frac{L(w)^+ b_ib_i^T L(w)^+}{1+\epsilon b_i^TL(w)^+b_i},
\end{align*}
and, hence thinking of $L(w)^+$ as a matrix-valued function of $w$,
\begin{align*}
\frac{\partial L(w)^+}{\partial w_i} &= \lim_{\epsilon \rightarrow 0} \frac{1}{\epsilon}\left[L(w)^+ -\epsilon \frac{L(w)^+ b_ib_i^T L(w)^+}{1+\epsilon b_i^TL(w)^+b_i} - L(w)^+\right]\\
&= -L(w)^+b_ib_i^T L(w)^+.
\end{align*}
Let $R_i$ denote the $i^{th}$ column of $R(w) = B L(w)^+ B^T$. By linearity:
\begin{align*}
\frac{\partial R}{\partial w_i} = - B \left (L(w)^+b_ib_i^T L(w)^+\right) B^T = -R_{i} R_{i}^T.
\end{align*}
\begin{align*}
&\text{Thus,} & &\frac{\partial F}{\partial w_i} = 2 \sum_{j \in \mathcal{S}} \left (\bar r(j) - [R(w)]_{j,j}\right ) \cdot [R(w)]_{i,j}^2,
\end{align*}
and so we obtain that the gradient equals 
$
\nabla F(w)= 2\left(R\circ R\right)\Delta(w).
$
\end{proof}

While gradient descent works well in our experiments, one may  also apply second order methods, which require $F(w)$'s Hessian. Using similar computations to those in Proposition~\ref{prop:gradient} we obtain: 

\begin{proposition} Let $I_{\mathcal{S}} \in \field{R}^{m \times m}$ be the diagonal matrix with a $1$ at each entry corresponding to $i \in \mathcal{S}$ and $0$s elsewhere and $\Delta(w)$  be as defined in Proposition~\ref{prop:gradient}. 
The Hessian matrix of $F(w)$ is:
\begin{align*}
H_F(w) = -4\left [R \diag (\Delta(w)) R \right ] \circ R + 2(R \circ R) I_{\mathcal{S}} (R \circ R).
\end{align*}
\end{proposition}

\spara{Acceleration via coordinate descent.}
Naively computing the gradient $\nabla F(w)$ via Proposition~\ref{prop:gradient} requires computing the full $m \times m$ matrix $R(w)$, which can be prohibitively expensive for large graphs -- recall that $m = {n \choose 2} = O(n^2)$. Note however, that the error vector $\Delta(w)$ only has nonzero entries at positions corresponding to the node pairs in $\mathcal{S}$. Thus, it suffices to compute just $|\mathcal{S}|$ columns of $R$ corresponding to these pairs, which can give a significant savings.

We obtain further savings using block coordinate descent. 
%
%
At each step we restrict our updates to a subset of edges $\mathcal{B} \subseteq [n] \times [n]$. Let $I_{\mathcal{B}}$ be the matrix with a $1$ at diagonal entries corresponding to elements of $\mathcal{B}$ and $0$'s elsewhere. We step in the direction of $I_{\mathcal{B}} \nabla F(w)$. Computing this step only requires forming the rows of $R$ corresponding to edges in $\mathcal{B}$. A  typical way to choose  $\mathcal{B}$ is at random. See Section~\ref{expsetup} for the actual implementation details.

\spara{Initialization.}
A good initialization for gradient descent can significantly accelerate the solution of Problem \ref{prob:effres2}. We use a strategy based on the exact solution to Problem \pone\ in Theorem \ref{thrm:profeffres}.

Since effective resistances form a metric, by triangle inequality, for any $u,v,w \in [n]$, $r_H(u,v) \le r_H(u,w) + r_H(w,v)$. Guided by this fact, given target resistances $\bar r(u,v)$ for $(u,v) \in \mathcal{S}$, we first ``fill in'' the constraint set. For $(w,z) \notin \mathcal{S}$, we set $\bar r(w,z)$ equal  to the shortest path distance in the graph $\bar{G}$ which has an edge for each pair in $\mathcal{S}$ with length $\bar r(u,v)$.

We thus obtain a  full set of target effective resistances. We can form $R$  with  $R_{u,v} = \bar r(u,v)$  and initialize the Laplacian of $H$ using the formula given in \eqref{eq:exactSolution} in Theorem \ref{thrm:profeffres}.
However, this formula is quite unstable and generally yields an output which is far from a graph Laplacian even when $R$ is corrupted by a small amount of noise. So we instead compute a regularized estimate, 
$$\tilde  L = -2 \cdot \left [ \left  (I - \frac{J}{n}\right)R\left(I-\frac{J}{n}\right) + \lambda I \right]^+,$$
where  $\lambda > 0$ can be chosen e.g. by line search. Generally, $\tilde  L$ will not be a valid graph Laplacian, but by removing negative edge weights, we typically obtain a good initialization for Problem \ref{prob:effres2}. 

\subsection{Graph Learning via Convex Optimization -- Problem \ref{prob:effres3}}  
 \label{sec:partial_reconstruct_noise} 
 We finally discuss how to efficiently solve our convex formulation, Problem~\ref{prob:effres3}.
We express this problem as a semidefinite program (SDP) which can be solved via a number of available packages.

We can re-express our effective resistance constraint as a positive semidefinite constraint using the Schur complement condition:
\begin{align*}
\xuv^T L^+ \xuv \leq  \bar r(u,v)  & \text{\hspace{1em} iff} & 
\begin{bmatrix}L & \xuv \\ \xuv^T & \bar r(u,v)\end{bmatrix} &\succeq 0 \text{ and }  L \succeq 0.
\end{align*}
Doing so yields the following program:

\vspace{.5em}
\begin{mdframed}[backgroundcolor=light-gray] 
\vspace{-.5em}
\begin{problem}[SDP Form of Problem~\ref{prob:effres3}]\label{prob:sdp}
Given vertex pairs $\mathcal{S} \subseteq [n] \times [n]$, and target effective resistance $\bar r(u,v)$ for every $(u,v) \in \mathcal{S}$,\\
 \vspace{-2.25em}
\begin{align*}
&\underset{L \in \mathcal{L}}{\text{minimize}} \tr(L)\\
& \text{subject to} \\
L \succeq 0 \text{ and }& \text{ $\forall\ (u,v)\in \mathcal{S}$, } \begin{bmatrix}L & \xuv \\ \xuv^T & \bar r(u,v)\end{bmatrix} \succeq 0
\end{align*}
\end{problem}
\end{mdframed}
\vspace{0.5em}
We require $L$  to be a valid graph Laplacian (i.e., constrain $L \in \mathcal{L}$) by adding linear constraints of the form:
\begin{align*}
\text{$\forall$ }i,\ L_{i,i} = - \sum_{j \neq i} L_{i,j}\text{\hspace{.25em} and $\forall$ }i \neq j,\ L_{i,j}  \le 0.
\end{align*}

\subsection{Extensions to Other Similarity Measures}\label{sec:extensions}

As discussed, our results generalize to random walk-based node similarities beyond effective resistances, such as personalized PageRank scores. Given localization parameter $\alpha \ge 0$, the personalized PageRank score $\pp_G(u,v)$ is the probability that a lazy random walk on graph $G$ which jumps back to $u$ with probability $\alpha$ in each step is at node $v$ in its stationary distribution \cite{page1999pagerank,haveliwala2003topic,andersen2006local}. 

Letting $W = \frac{1}{2}(I + AD^{-1})$ be the lazy random walk matrix, $\pp_G(u,v)$ is the $v^{th}$ entry of the personalized PageRank vector:
\begin{align}\label{eq:pp}
\pp_G(u) = \alpha (I - (1-\alpha) W)^{-1} e_u.
\end{align}
This vector gives the stationary distribution for the random walk on $G$ and thus
the personalized  PageRank $\pp_G(u,v)$ is its $v^{th}$ entry.

It is not hard to show an analogous result to Theorem \ref{thrm:profeffres}, that given a full set of exact personalized PageRank scores, full recovery of $G$ is possible. Roughly, if we let $P$ be the matrix with $\pp_G(u)$ as its $u^{th}$ column, we have $P = \alpha (I - (1-\alpha) W)^{-1} $ and can thus solve for the random walk matrix $W$, and the graph $G$. This gives:
\begin{theorem}\label{thm:pp} For any connected graph $G$, given personalized PageRank score $\pp_G(u,v)$ for each $(u,v) \in [n] \times [n]$, there is algorithm returning $G$ (up to a scaling of its edge weights) in $O(n^3)$ time.
\end{theorem}
Further, it is possible to formulate a problem analogous to Problem \ref{prob:effres2} and solve for a graph matching a subset of personalized PageRank measurements as closely as possible. As shown in Figure \ref{fig:ppvseffres}, personalized PageRank often gives a stronger signal of global graph structure than effective resistance. 
\begin{figure}[h!]
\centering
\includegraphics[width=.95\textwidth]{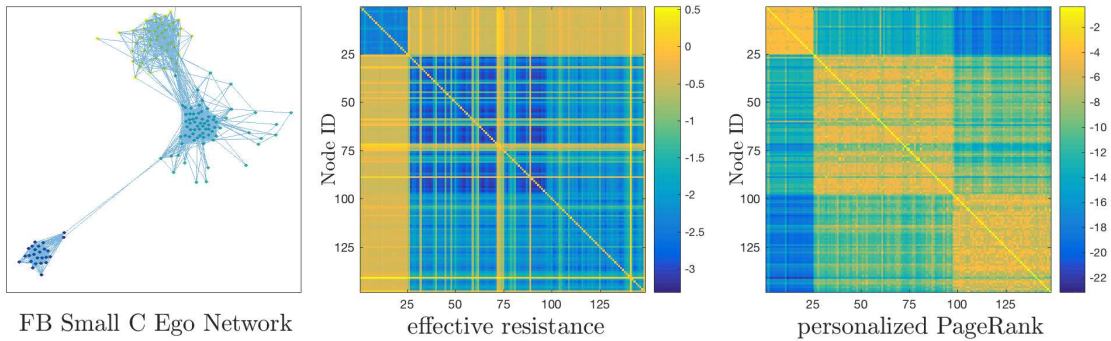} 
\caption{Personalized PageRank   correlates better  than commute times with the cluster structure in the {\sc FB Small C} network (see Table \ref{tab:datasets}).   
Heatmaps are shown in log scale.}
\label{fig:ppvseffres}
\end{figure}
To create the plot, nodes are sorted by their value in the Laplacian Fielder vector, which corresponds roughly to residence in different clusters.  In an extended version of our work, we will provide detailed empirical results with personalized PageRank, and other random walk measures.

\section{Empirical results}
\label{sec:exp}
In this section we present an in depth experimental study of how well our methods can learn a graph given  
a set of (noisy) effective resistance measurements. We seek to answer two key questions:

\begin{enumerate}
\item Given a set of effective resistance measurements, can we find a graph matching these measurements via the optimization formulations posed in Problems \ref{prob:effres2} and \ref{prob:effres3} and the corresponding algorithms discussed in Sections \ref{sec:full_reconstruct_noise} and \ref{sec:partial_reconstruct_noise}?
\item What structure does the graph we learn via our optimization approach share with the underlying network that produced the resistance measurements? Can it be used to predict links in the network? Does it approximately match the network on effective resistances outside the given set of measurements, or share other global structure?
\end{enumerate}

We address these questions by examining a variety of synthetic, and  social network graphs.

\subsection{Experimental Setup} 
\label{expsetup}
Table~\ref{tab:datasets} lists the networks analyzed in our experiments. These include two synthetic examples: an $8 \times 8$ two dimensional grid graph and a $k$-nearest neighbor graph constructed for vectors drawn from a Gaussian mixture model with two clusters. 
The other networks are Facebook `ego networks' obtained from the Stanford Network Analysis Project (SNAP) collection \cite{snapnets,leskovec2012learning}. Each of these networks is formed by taking the largest connected component in the social circle of a specific user (whose nodeId is shown in Table \ref{tab:datasets}). 

\begin{table}[!ht]
\centering
\begin{tabu}{|l|c|c|} \hline
Name  & \# of nodes, $n$  & \# of edges, $m$ \\
 \tabucline[1pt]{1-3}
{\sc Grid (synthetic)} &  64 & 224 	 \\
{\sc  k-nn (synthetic) }  & 80 & 560  \\ \hline
{\sc FB Small A \footnotesize(nodeId 698)}   & 40	& 220 \\ 
{\sc FB Small B \footnotesize(nodeId 3980)}   & 44	& 138	 \\ 
{\sc FB Small C \footnotesize(nodeId 414)}   & 148	& 1692 \\
{\sc FB Small D \footnotesize(nodeId 686)}   & 168	& 1656	 \\  \hline
{\sc FB Medium A \footnotesize(nodeId 348)}   & 224	& 3192  \\  
{\sc FB Medium B \footnotesize(nodeId 0)}   & 324	& 2514 	 \\  \hline
{\sc FB Large A \footnotesize(nodeId 3437)}   & 532	& 4812  \\  
{\sc FB Large  B \footnotesize(nodeId 1912)}   & 795	& 30023 	 \\  \hline 
\end{tabu}
\vspace{1em}
\caption{\label{tab:datasets}  Datasets for experiments. {\sc FB} denotes ``Facebook''.}
\end{table}

For all experiments, we provide our algorithms with effective resistances that are uniformly sampled from the set of all ${n \choose 2}$ effective resistances. We  sample  a fixed fraction $f \myeq \frac{|\mathcal{S} |}{{n \choose 2}} \times 100\%$ of all possible measurements. We typically use $f \in \{10 , 25, 50, 100 \}\%$. In some cases, these resistances are corrupted with i.i.d. Gaussian noise $\eta \sim \mathcal{N}(0,\sigma^2)$. We experiment with different values of variance $\sigma^2$.

For Problem \ref{prob:effres2} we implemented gradient decent based on the closed-from gradient calculation in Section \ref{sec:full_reconstruct_noise}. Line search was used to optimize step size at each iteration since it significantly outperformed implementations with fixed step sizes. For larger problems, block coordinate descent was used as described in Section \ref{sec:full_reconstruct_noise}, with the coordinate set chosen uniformly at random in each iteration. We set the block size $|\mathcal{B}|=5000$. For Problem \ref{prob:effres3} we used MOSEK convex optimization software, accessed through the CVX interface \cite{mosek,cvx}. All experiments were run on a computer with a 2.6 GHz Intel Core i7 processor and 16 GB of main memory.

\subsection{Learning Synthetic Graphs}
\label{sec:learnsynthetic}

We first evaluate our graph learning algorithms on {\sc Grid} and {\sc k-nn}, which are simple synthetic graphs with clear structure. 

\spara{Least squares formulation.} We first observe that gradient descent effectively minimizes the objective function of Problem \ref{prob:effres2} on the {\sc Grid} and {\sc k-nn} graphs. We consider the normalized objective for a constraint set $\mathcal{S}$ and output graph $\optG$:
\begin{align}\label{normalizedObj}
\wh F(\optG) = \frac{\sum_{(u,v) \in \mathcal{S}} \left [{r}_\optG(u,v)-\bar r(u,v)\right]^2}{\sum_{(u,v) \in \mathcal{S}} \bar r(u,v)^2}.
\end{align}
For noise variance 0, $\min_\optG \wh{F}(\optG) = 0$ and in Figure \ref{gridErrorReduce0noise} we see that for {\sc Grid} we in fact find $H$ with $\wh{F}(H) \approx 0$ for varying sizes of $\mathcal{S}$. Convergence is notably faster when $100\%$ of effective resistances are included in $\mathcal{S}$, but otherwise does not correlate strongly with the number of constraints.

 \begin{figure} 
\begin{tabular}{cc}
        \includegraphics[width=0.45\textwidth]{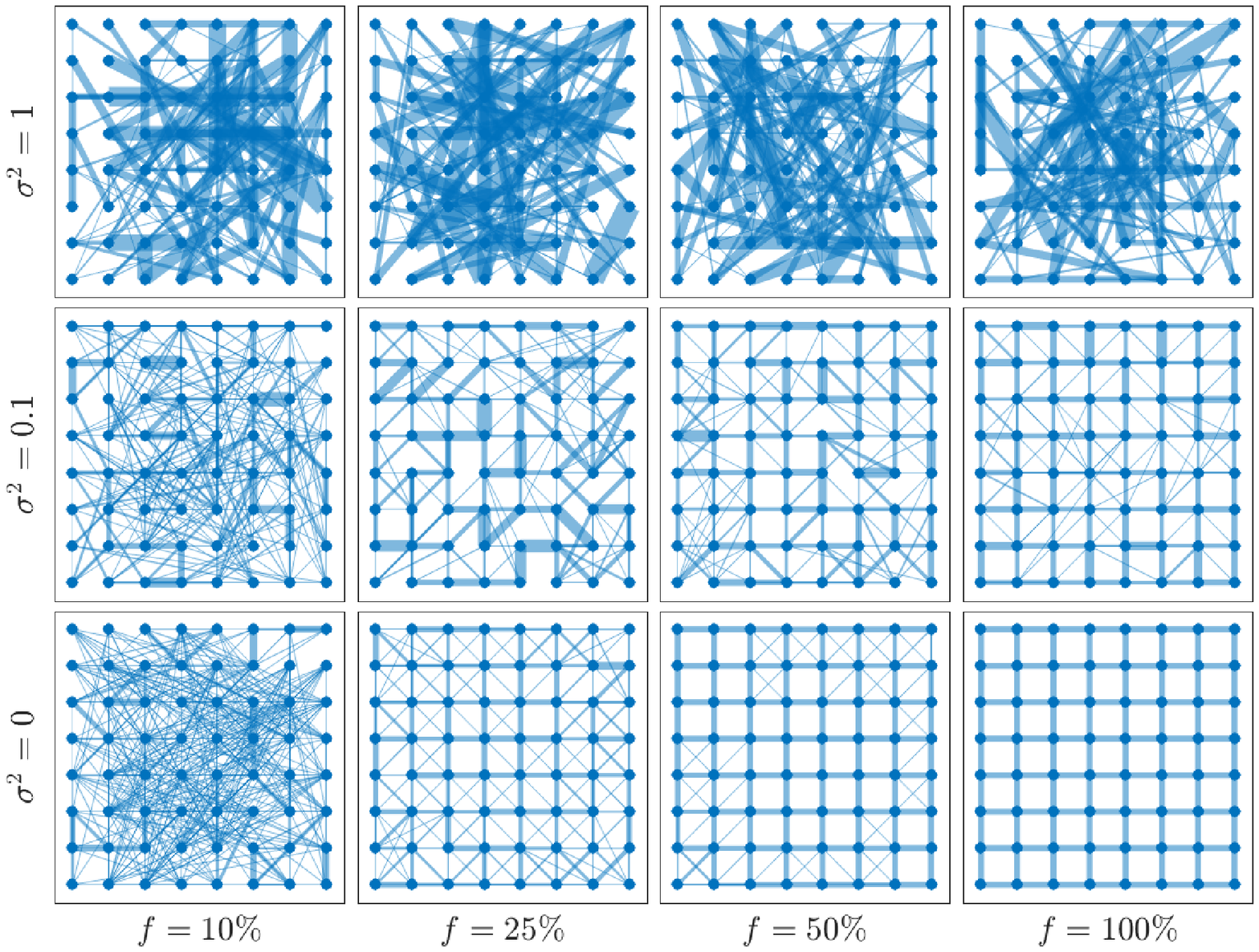} & 
        \includegraphics[width=0.45\textwidth]{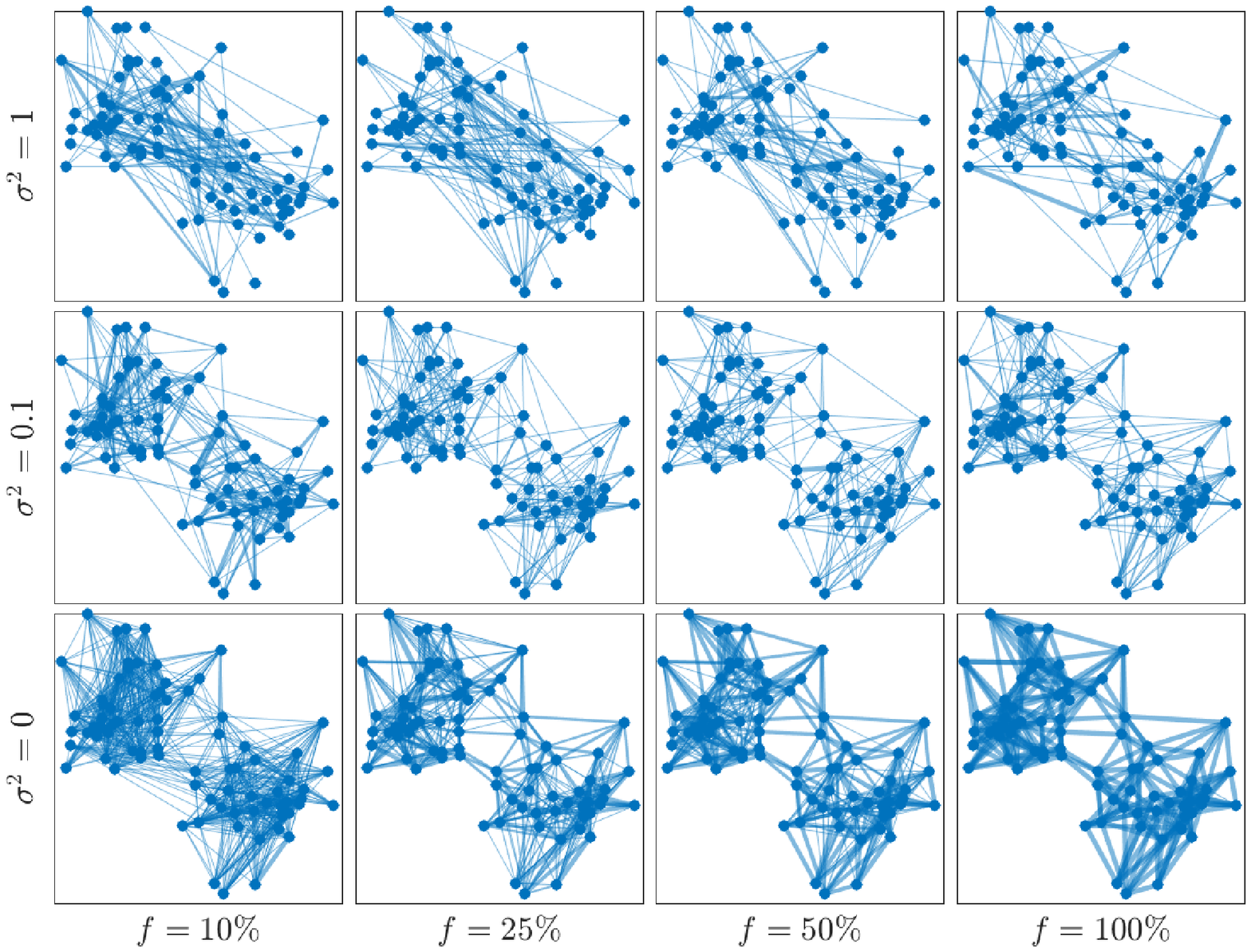} 
\end{tabular}
\caption{\label{GDvisualRecovery} Graphs learned by solving Problem \ref{prob:effres2} with gradient descent run to convergence for uniformly sampled effective resistances with varying levels of Gaussian noise. Edge width is proportional to edge weight in the plots.}
\end{figure}

In Figure \ref{gridErrorReduce0noise} we also plot the \emph{generalization error}:
\begin{align}
\label{generalization_error}
F_{gen}(\optG) = \frac{\sum_{(u,v) \in [n] \times [n]} \left [{r}_\optG(u,v)- r_G(u,v)\right]^2}{\sum_{(u,v) \in [n] \times [n]} {r}_G^2(u,v)},
\end{align}
where $r_G(u,v)$ is the true effective resistance, uncorrupted by noise.
$F_{gen}(\optG)$ measures how well the graph obtained by solving Problem \ref{prob:effres2} matches \emph{all} effective resistances of the original network. We confirm that generalization decreases with improved objective function performance, indicating that optimizing Problem \ref{prob:effres2} effectively extracts network structure from a small set of effective resistances. We observe that the generalization error is small even when $f=10\%$, and becomes negligible as we increase the fraction $f$ of measurements, even in the presence of noise.

\begin{figure*}[!t]
\begin{center}
\begin{tabular}{cc}
        \includegraphics[width=0.4\textwidth]{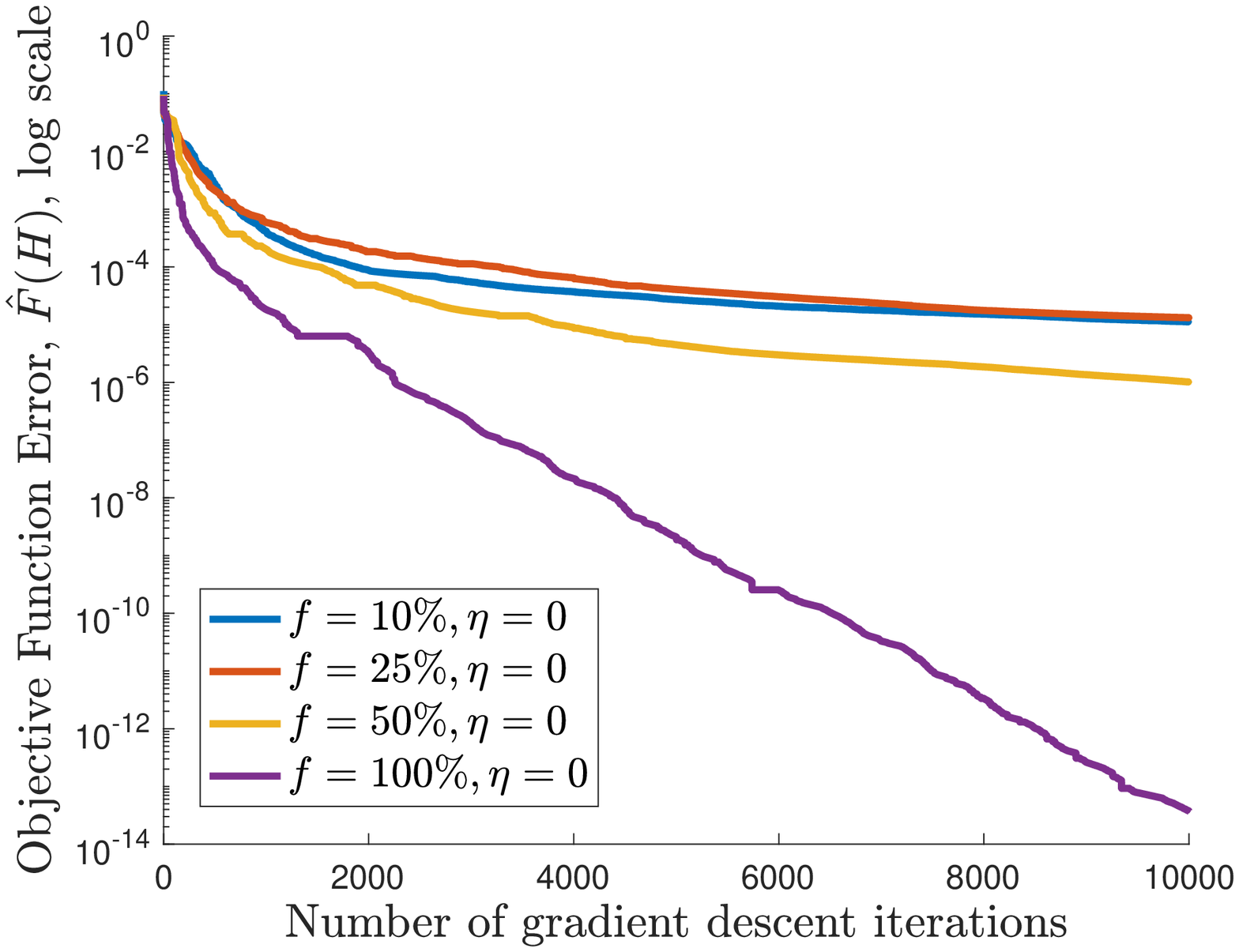} &        \includegraphics[width=0.4\textwidth]{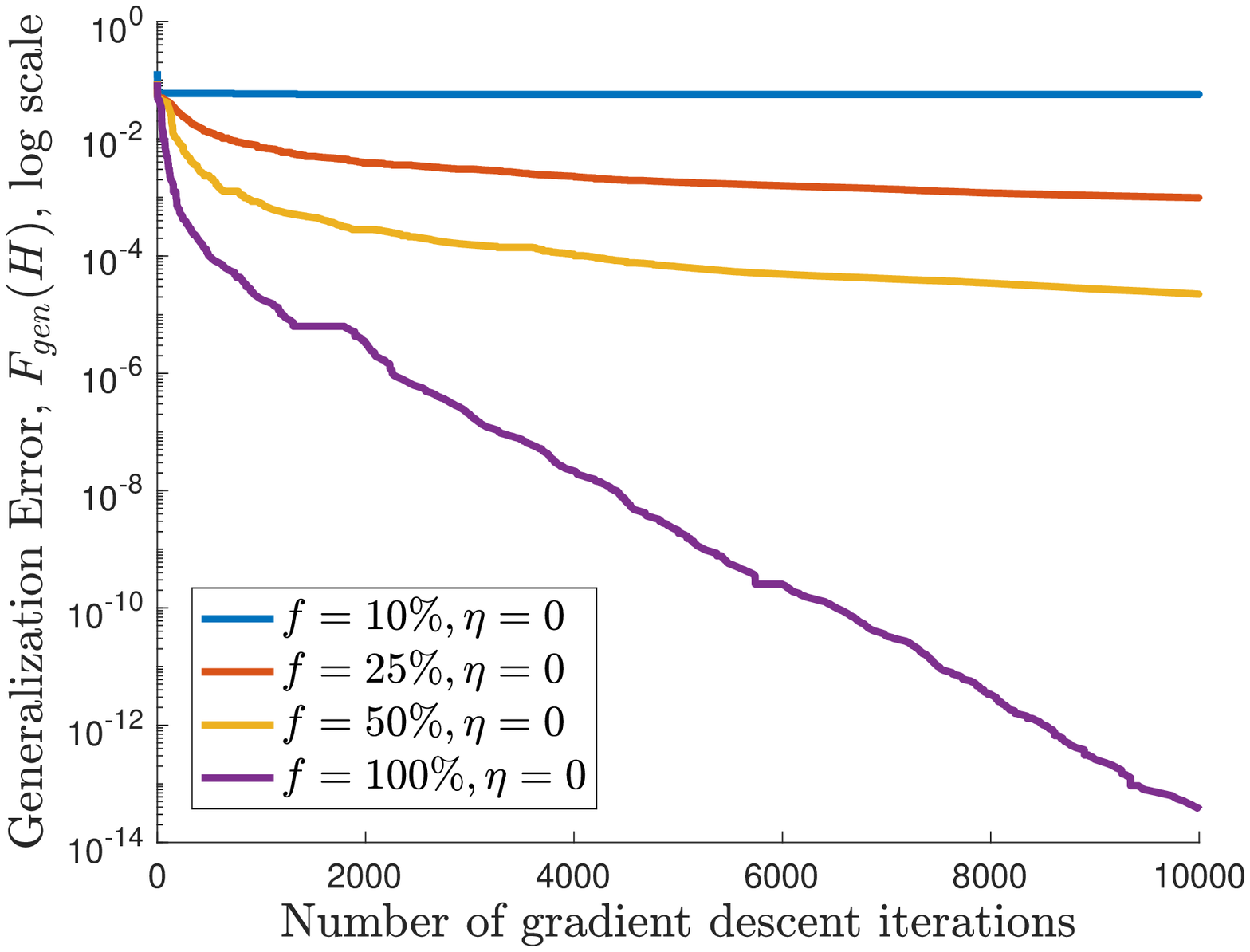} \\ 
                \includegraphics[width=0.4\textwidth]{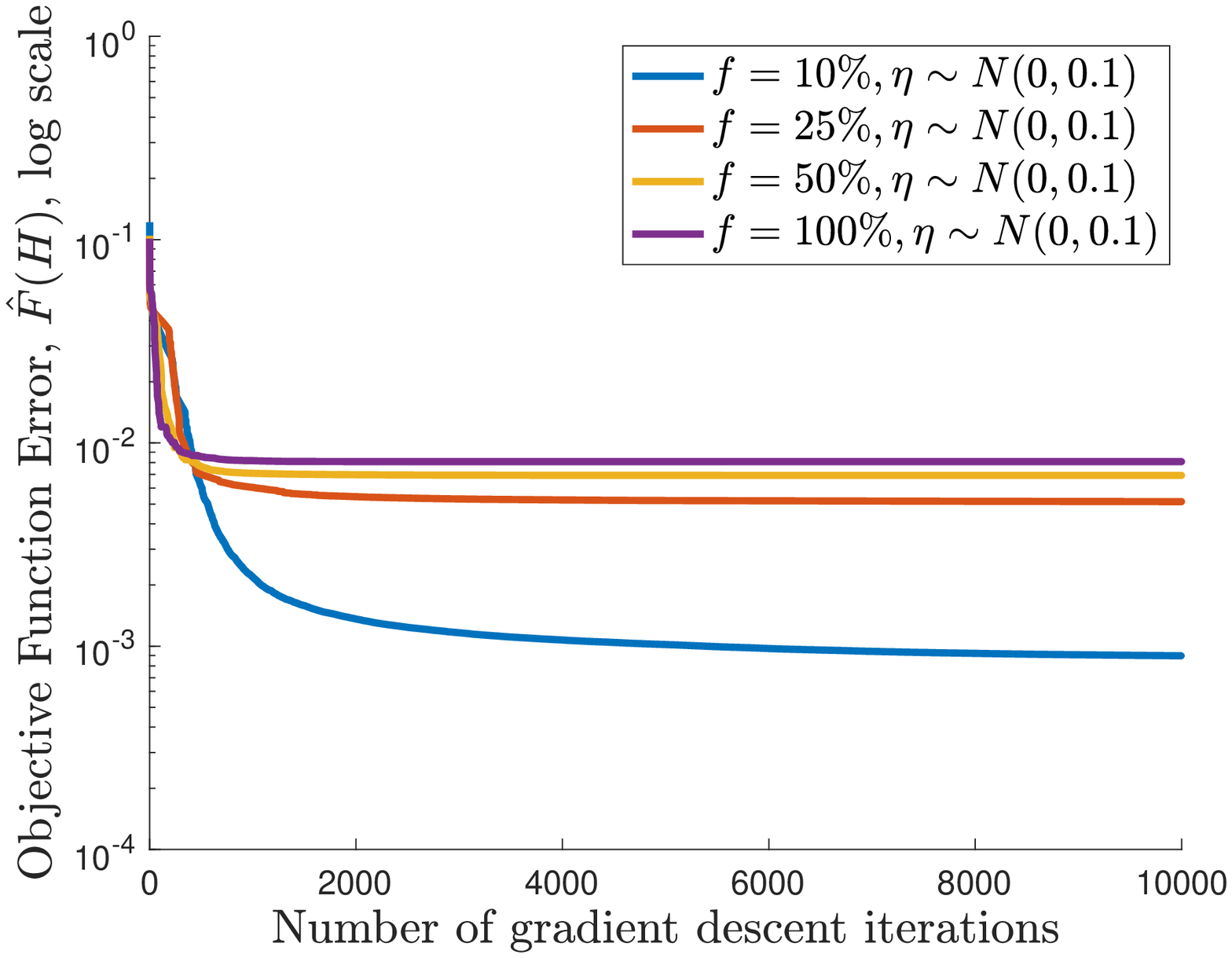} &         \includegraphics[width=0.4\textwidth]{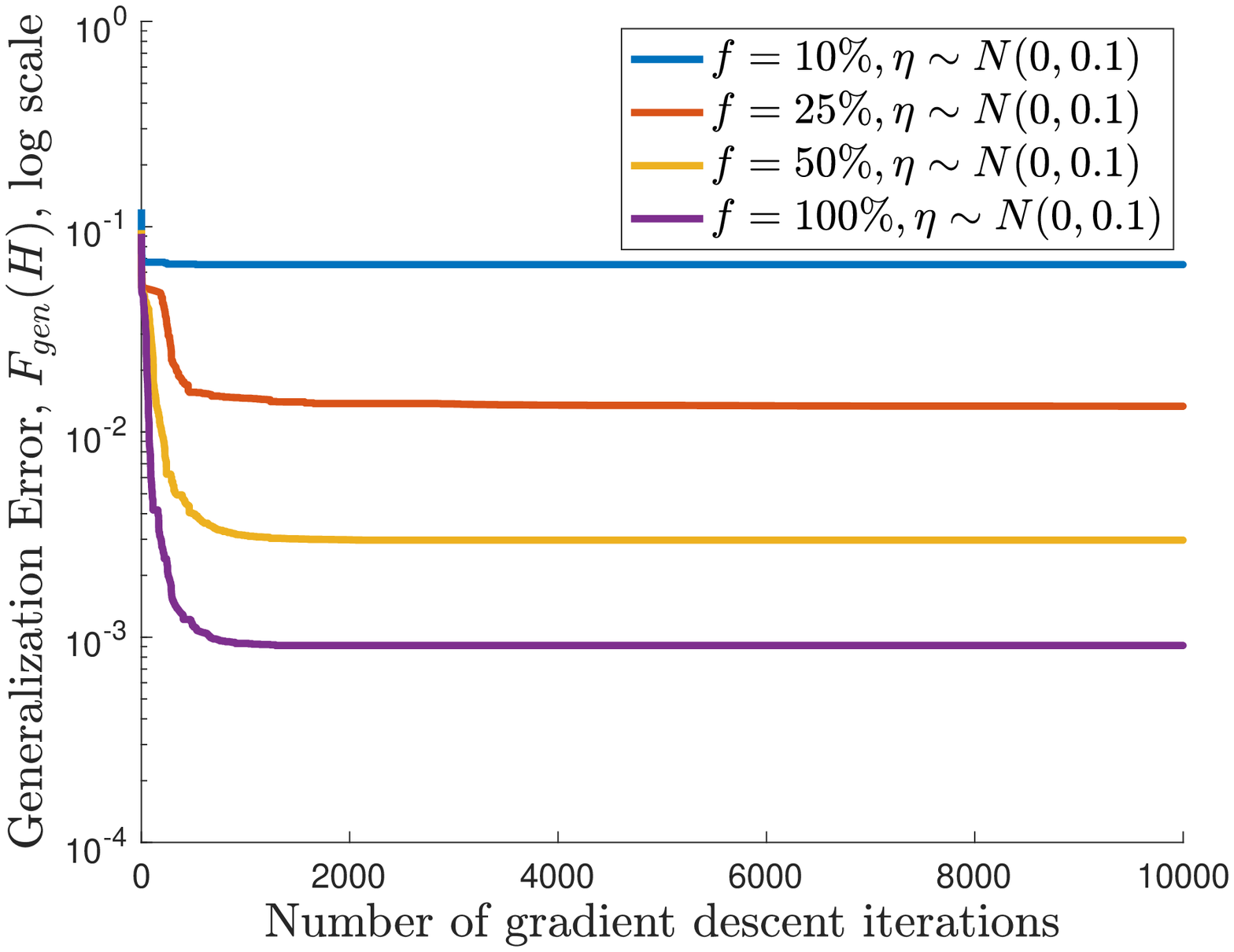} 
\end{tabular}
\end{center}
\caption{  \label{gridErrorReduce0noise}  Objective error and generalization error for Problem \ref{prob:effres2}, as defined in \eqref{generalization_error} for {\sc Grid}. For details, see Section~\ref{sec:learnsynthetic}.}
\end{figure*}

We repeat the same experiments with Gaussian noise added to each resistance measurement. The variance of the noise, $\sigma^2$, is scaled relatively to the mean effective resistance in the graph, i.e., we set $ \bar r(u,v) = r_G(u,v) + \mathcal{N}(0,\bar{\sigma}^2)$ where: 
\begin{align}
\label{variance_scale}
\bar{\sigma}^2 = \frac{\sigma^2}{{n \choose 2}} \cdot \sum_{(u,v) \in [n] \times [n]} {r}_G(u,v).
\end{align}
While generally $\min_\optG \wh{F}(\optG) > 0$ when $\bar {r}(u,v)$ is noisy   (it is likely that there is no graph consistent with these noisy measurements), we see that the objective value still decreases steadily with a larger number of iterations. Generalization error also decreases as desired.

 \begin{figure*}[!t]
 \begin{center}
\begin{tabular}{cc}
        \includegraphics[width=0.4\textwidth]{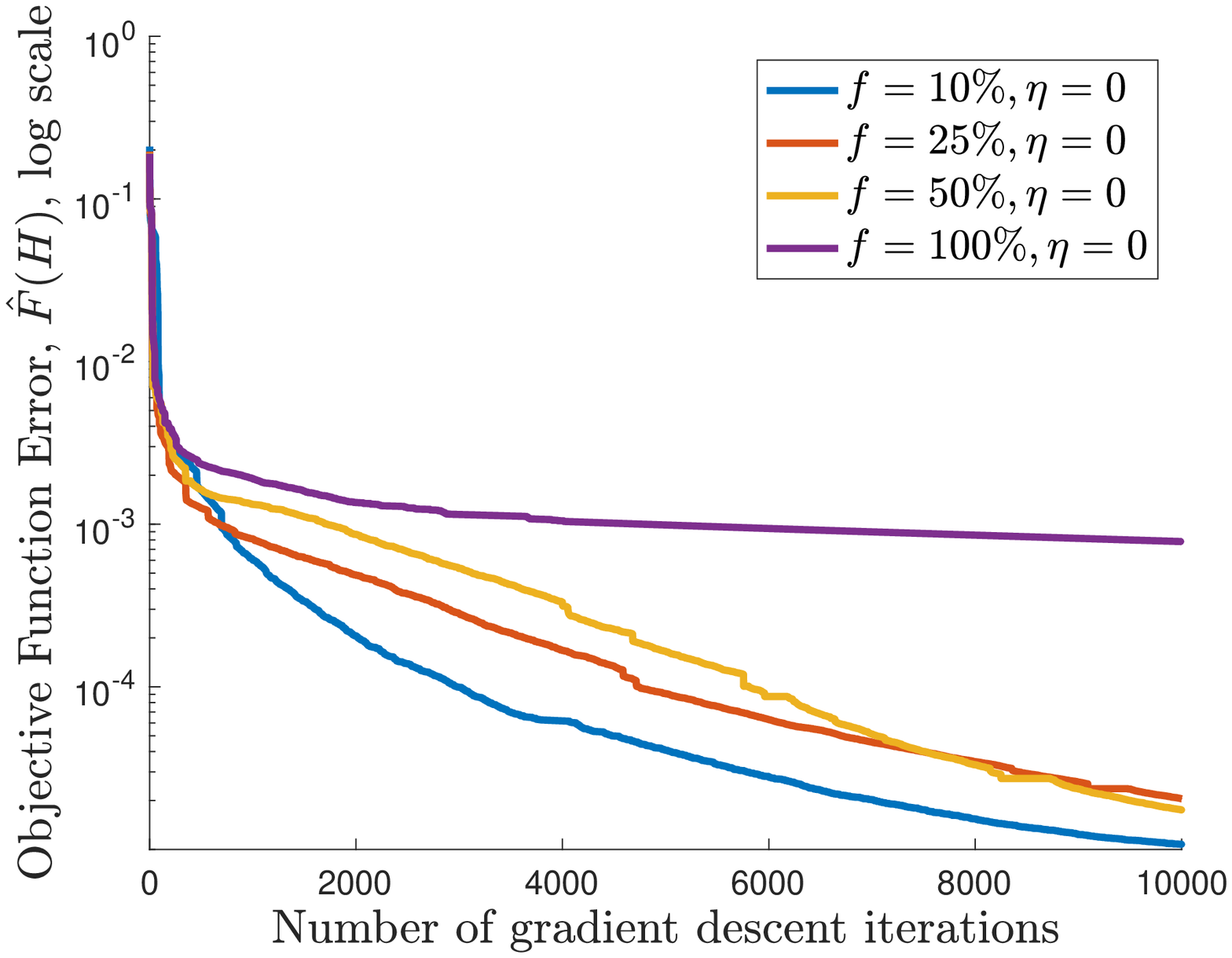} &        \includegraphics[width=0.4\textwidth]{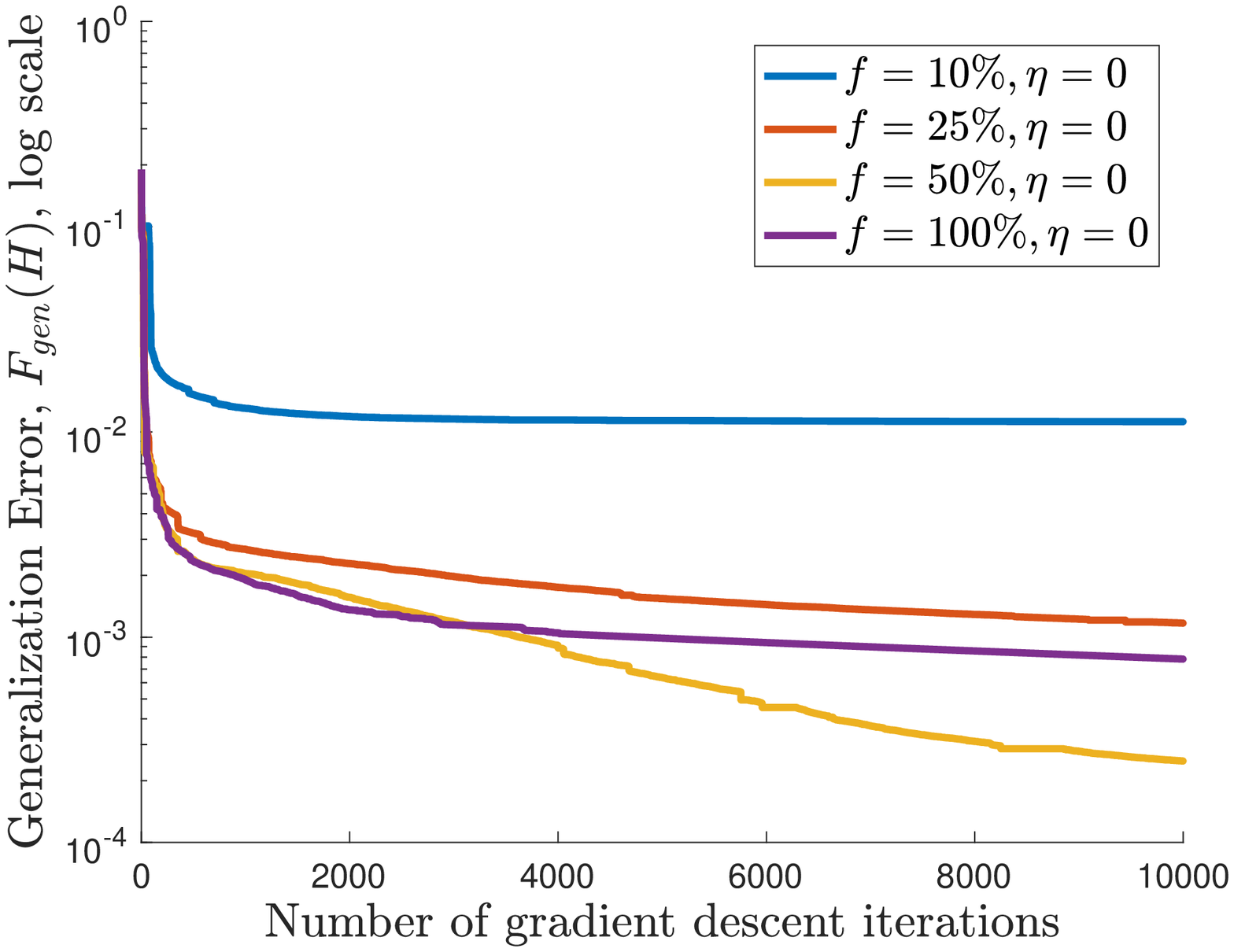} \\ 
                \includegraphics[width=.4\textwidth]{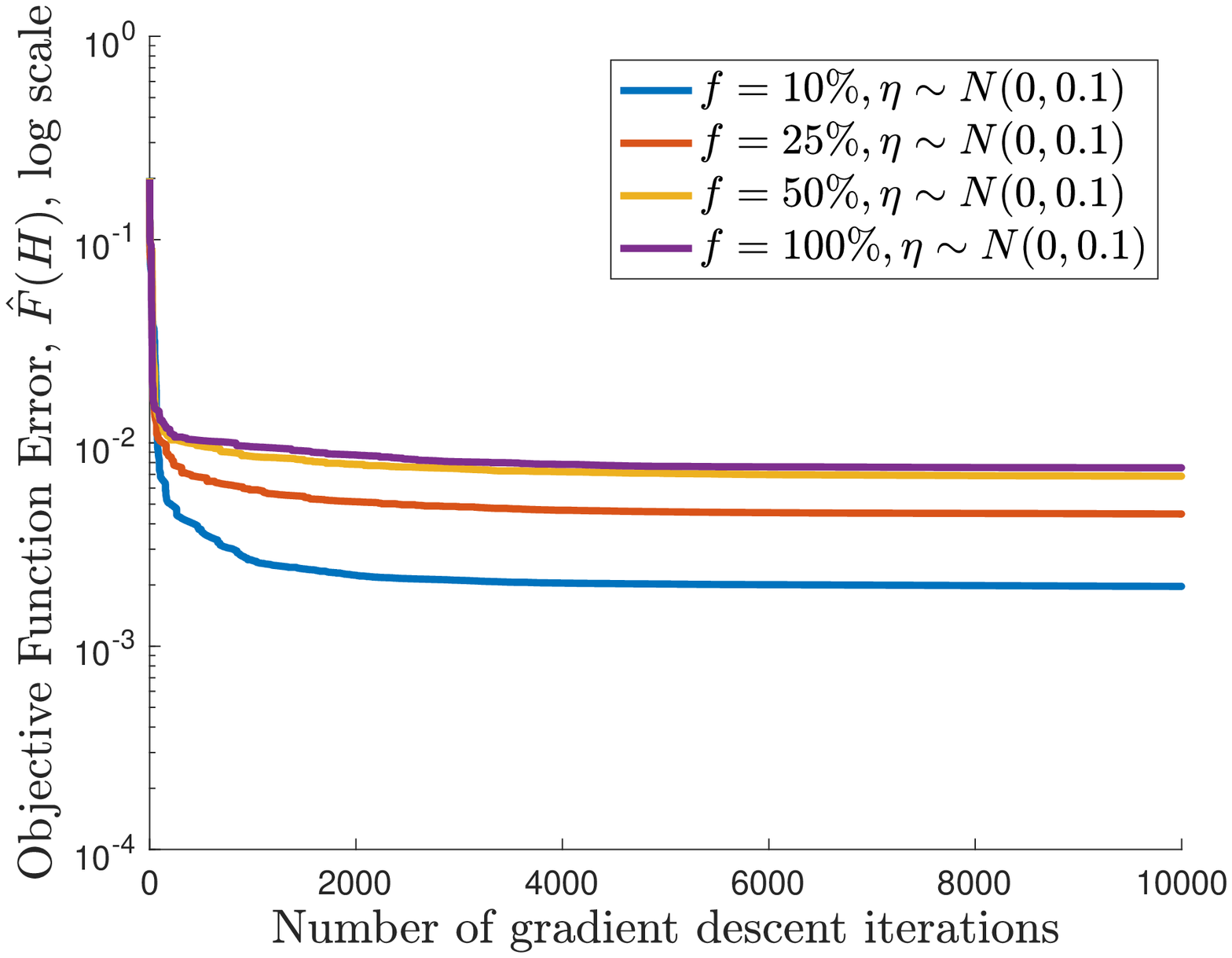} &         \includegraphics[width=0.4\textwidth]{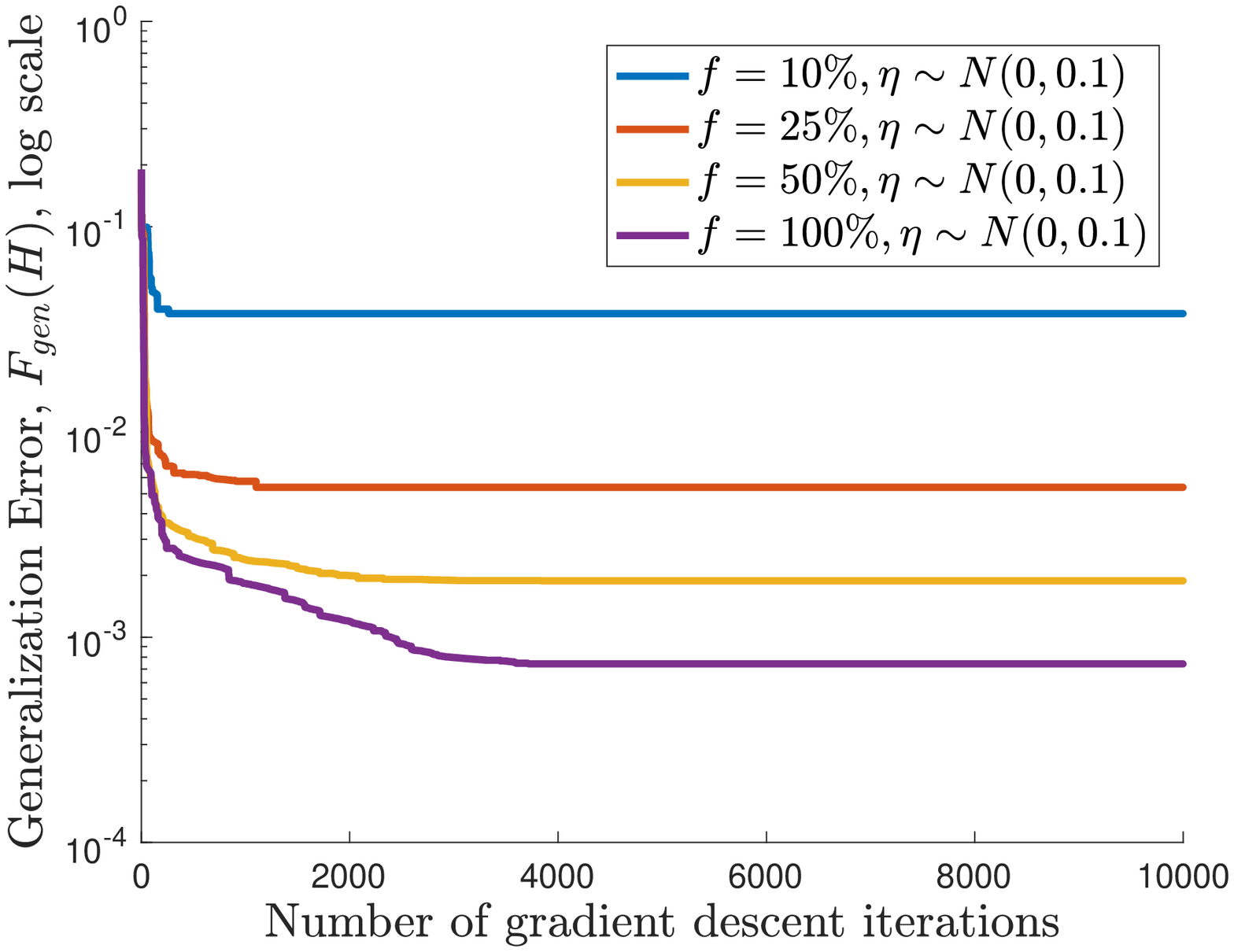}
\end{tabular}
\end{center}
\caption{  \label{clusterErrorReduce} Objective error and generalization error for Problem \ref{prob:effres2}, as defined in \eqref{generalization_error} for {\sc k-nn}. Objective function error decreases steadily, leading to improved generalization error.}
\end{figure*}

We obtain similar results by applying Problem \ref{prob:effres2} to a  $k$-nearest neighbor graph with two clear clusters of data points. Again, it is apparent in Figure \ref{clusterErrorReduce} that gradient descent converges easily for a variety of noise levels and constraint sets. Convergence on the least squares objective leads to improved generalization error. 

Figure \ref{GDvisualRecovery} shows the graphs obtained from solving the problem for varying $\sigma^2$ and $f$. For both graphs, when $\sigma^2 = 0$ and $f=100\%$, the original network is recovered exactly.
Reconstruction accuracy decreases with increasing noise and a decreasing number of constraints. For {\sc Grid}, even with $25\%$ of constraints, nearly full recovery is possible for $\sigma^2 = 0$ and recovery of approximately half of true edges is possible for $\sigma^2 =0.1$. For {\sc k-nn}, for $\sigma^2 = 0$ and $\sigma^2 = 0.1$ we observe that cluster structure is recovered.  Detailed quantitative results for both networks are given in Table \ref{tab:results}.

\begin{table}
\centering
\begin{tabu}{|c|c|[1pt]c|c|} \hline
 \multicolumn{2}{|c|[1pt]}{{\sc Grid}, 64 nodes} &  \multicolumn{2}{c|}{{\sc  k-nn}, 80 nodes} \\\hline
  \# of constraints & Time (min.) & \# of constraints & Time (min.)\\
   \tabucline[1pt]{1-4} 
202 & 238.8 & 316 & 1216.0	 \\
504 &761.6  & 790 &3673.2 	 \\
1008 & 1309.5 & 1580 & 8008.5 \\
2016 & 2183.7 & 3160 & 16192.1\\
\hline
\end{tabu}
\caption{\label{tab:sdptimes}Semidefinite program (SDP) optimization for Problem \ref{prob:effres3}. Runtime is averaged over noise levels $\sigma^2= 0, 0.1, 1$. }
\end{table}
 
\spara{Convex formulation.} We next evaluate the performance graph learning via the convex relaxation in Problem \ref{prob:effres3}. In this case, we do not focus on convergence as we solve the problem directly using a semidefinite programming (SDP) routine. 
Unlike for Problem \ref{prob:effres2}, solving Problem \ref{prob:effres3} does not recover the exact input graph, even in the noiseless all pairs effective resistance case. This is because the input graph does not necessarily minimize the objective of Problem \ref{prob:effres3} since there can be other graphs with smaller total edge weight and lower effective resistances.

\begin{table*}[h]
\begin{footnotesize}
\begin{tabu}{|c|c|c|c|c|c|c|c|c|c|} \hline
\multirow{2}{*}{Network}  & \multirow{2}{*}{Algorithm} &  \multirow{2}{*}{\specialcell{$\sigma^2$}} & \multicolumn{2}{c|}{\specialcell{Objective \\ function error}} & \multicolumn{2}{c|[1pt]}{\specialcell{Effect. resistance \\ generalization  error}} & \multirow{2}{*}{\specialcell{\% Edges \\ learned \\ baseline\vspace{.85em}}} &\multicolumn{2}{c|}{\specialcell{\% Edges \\ learned} } \\ \cline{4-7}\cline{9-10}
& & & $f=10\%$ & $f=25\%$ & $f=10\%$ & $f=25\%$ & & $f=10\%$ & $f=25\%$\\
 \tabucline[1pt]{1-10}
{\sc Grid } & GD & 0 & $.00001$ & $.00001$ &  $.06559$& $.00099$& 5.56 	& 20.54 & 88.39\\
 & GD & .1 & $.00090$ & $.00514$ & $.08129$& $.01336$ &  	& 25.89 & 50.00 \\
 & SDP & 0 & na & na & $.08758$& $.07422$ & 	&16.07 & 25.00 \\
 & SDP & .1 & na & na & $.09549$& $.09343$ & 	&12.50 & 26.79 \\ \hline
{\sc  k-nn} & GD & 0 & $.00001$ & $.00002$ & $.01122$&$.00117$ & 11.58 	& 44.54& 72.68\\
& GD & .1 & $.00197$ & $.00447$ & $.05536$&$.00709$ &  	& 25.96& 41.53\\
 & SDP & 0 & na & na & $.09314$& $.10399$ & 	&27.05 & 48.36 \\ 
 & SDP & .1 & na & na & $.11899$ & $.14097$& 	& 24.32 & 39.89 \\ \hline
{\sc FB Small A} &GD  & 0	& $.01345$& $.00001$ & $.21097$ & $.00984$ & $28.20$&  $44.54$ & $75.00$\\ 
& GD & .1	& $.00017$ & $.00204$& $.07964$& $.01687$  &  & $53.64$ &$60.00$\\ \hline
{\sc FB Small B} & GD  & 0	& $.00002$	& $.00003$ & $.01515$& $.00623 $& $14.59$ &$42.75$ & $48.55$ \\
 & GD  & .1	& $.00032$&$.00206$ &$.02229$&$.01291$	 &&$36.23$ &$43.48$ \\ \hline
{\sc FB Small C}  & GD & 0	& $.00162$&	$.00166$ & $.00217$&$.00203$& $15.55$&$57.03$ & $59.51$ \\
 & GD & .1	& $.00532$&	$.00644$ & $.01542$& $.00218$  &  &$52.66$ & $57.51$\\\hline
{\sc FB Small D} & GD  & 0	& $.00335$& $.00434$& $.00821$&$.00830$	  & $11.80$  &$21.92$ &$24.52$\\
& GD  & .1	& $.00610$& $.18384$& $.00923$&$.21426$	  &  &$21.38$ &$21.20$\\ \hline
{\sc FB Medium A} & GD  & 0	& $.00447$&	$.00665$ &$.02910$&$.01713$&$12.78$	  &$23.50$ & $25.59$  \\  \hline
{\sc FB Medium B} & CD  & 0	& $.00224$&$.01255$&$.00862$& $.01471$	 &$4.80$&$18.97$ & $22.15$\\
& CD  & .1	& $.01174$&$.03182$	&$.01687$ & $.03295$ &&$17.50$& $16.03$\\
 \hline
{\sc FB Large A}  & CD& 0	& $.00516$&	$.00796$ & $.00682$& $.00862$& $3.41$&$10.52$ & $12.45$\\  \hline
{\sc FB Large  B}  & CD & 0	& $.00524$&	$.00440$  & $.00635$& $.00580$&$9.51$&$20.26$ & $24.95$ \\ 
  & CD & .1	& $.12745$&	$.34646$ &$.14532$ & $.36095$ && $19.43$ & $16.97$ \\ \hline  
\end{tabu}
  \caption{\label{tab:results}Graph recovery results. All results use a randomly sampled subset of $f=10\%$ or $f=25 \%$ of all effective resistances.
  For ``Algorithm'', GD denotes gradient descent and CD denotes block coordinate descent with random batches of size $5000$. ``Noise level, $\sigma^2$'' indicates that the target resistances were set to $\bar r(u,v) = r_G(u,v) + \mathcal{N}(0,\sigma^2 \cdot mean_{u,v}(r_G(u,v)))$.  
  ``\% Edges baseline'', is the edge density of the underlying network, equivalent to the expected edge prediction accuracy achieved with random guessing.}
\end{footnotesize}
\end{table*}

However, the learned graphs \emph{do capture} information about edges in the original: their heaviest edges typically align with true edges in the target graph. This property is captured in the quantitative results of Table \ref{tab:results}. Qualitatively, it is very apparent for {\sc Grid}: in  Figure~\ref{SDPvisualRecovery} we mark the 224 heaviest edges in the learned graph in red and note that this set converges exactly on the grid. 

The convex formulation never significantly outperforms the least squares formulation, and significantly underperforms for small constraint sets.
Additionally, the semidefinite program scales poorly with the number of nodes and effective resistance constraints. Sample runtimes are included in Table \ref{tab:sdptimes}.
Due to these considerations, we use the least squares formulation of Problem \ref{prob:effres2} instead of the convex formulation in our experiments on real social networks. However, we believe there is further opportunity for exploring Problem \ref{prob:effres3}, especially given its provable runtime guarantees.  

\begin{figure}[h!]
\centering
\includegraphics[width=0.65 \textwidth]{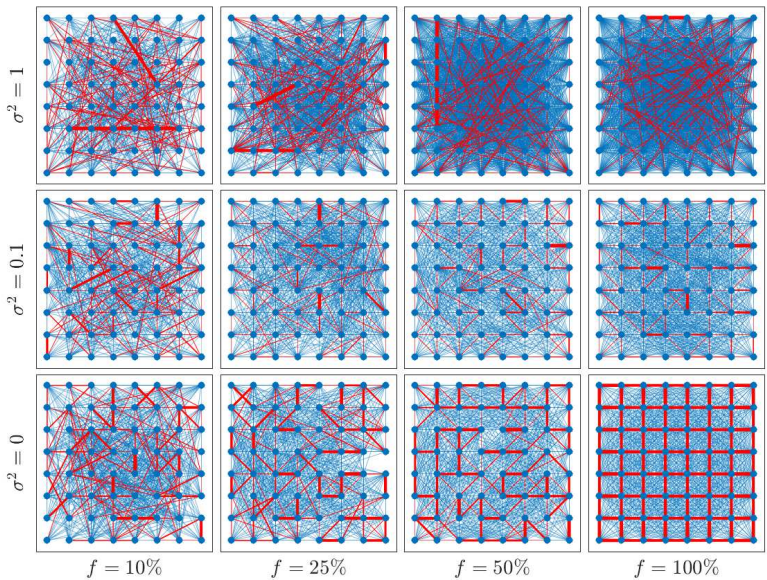}
\caption{\label{SDPvisualRecovery} Graphs learned from solving the convex program in Problem \ref{prob:effres3} for uniformly sampled effective resistances from {\sc Grid} with varying $f,\sigma^2$. Heaviest edges marked in red.}
\end{figure}

\normalsize
\subsection{Learning  Social Network Graphs}

We conclude by demonstrating the effectiveness of the least squares formulation of Problem \ref{prob:effres2} in learning Facebook ego networks from limited effective resistance measurements. We consider three metrics of performance, shown in Table \ref{tab:results} for a number of networks learned from randomly sampled subsets  of effective resistances, corrupting with varying levels of noise.
\begin{enumerate}
\item Objective Function Value: the value of the objective function of Problem \ref{prob:effres2}, normalized as described in \eqref{normalizedObj}. 
\item Generalization Error: the error in reconstructing the full set of effective resistances of the true graph, as defined in \eqref{generalization_error}.
\vspace{-2mm}
\item Edges Learned: the rate of recovery for edges in the true graph. We utilize a standard metric from the link prediction literature \cite{kleinberg1999web}: given underlying graph $G$ with $m$ edges and learned graph $H$, we consider the $m$ heaviest edges of $H$ and compute the percentage of $G$'s edges contained in this set.
\end{enumerate}

\spara{Results.} We find that as for the synthetic {\sc Grid} and {\sc k-nn} graphs, we can effectively minimize the objective function of Problem \ref{prob:effres2} for the Facebook ego networks. Moreover, this minimization leads to very good generalization error in nearly all cases. i.e., we can effectively learn a graph matching our input on all effective resistances, even when we consider just a small subset.  

For all graphs, we are able to recover a significant fraction of edges, even when just considering $10\%$ or $25\%$ of effective resistance pairs. We obtain the best recovery for small graphs, learning over half of the true edges in {\sc FB Small A} and {\sc FB Small C}. Even for larger graphs, we can typically recover over $20\%$ of the true edges.  

Typically, the number of edges learned increases as we increase the number of constraints and decrease the noise variance. However, occasionally,  considering fewer effective resistances in fact improves learning, possibly because we more effectively solve the underlying optimization problem.\vspace{1em}

 \section{Conclusion}
\label{sec:concl}
In this work, we show  that a small subset of noisy effective resistances can be used to learn significant information about a network, including predicting a large percentage of edges and recovering global structure, such as accurate approximations to all pairwise effective resistances. From a privacy standpoint, our results raise major concerns about  releasing random walk-based pairwise node similarity information as it entails significant privacy risk. From a data mining perspective, our methods can be used for graph mining, even when computing all effective resistances exactly is infeasible. 

Our work leaves  a number of future research directions open, including giving a provable polynomial time algorithm for the least squares formulation (Problem \ref{prob:effres2}), extending our work to other similarity metrics, and scaling our methods to larger social networks.

\bibliography{ref}
\bibliographystyle{abbrv}

\end{document}